\documentclass[11pt,reqno]{amsart}
\usepackage{amsmath}
\usepackage{txfonts}
\usepackage{mathrsfs}
\usepackage{amsfonts}
\allowdisplaybreaks[4]
\usepackage{graphicx} 
\usepackage{epsfig}
\usepackage{bbm}
\usepackage{amssymb}
\usepackage{cite,color}
\usepackage{tikz}
\usepackage{esint}
\usepackage{bm}
\newtheorem{theorem}{Theorem}

\newtheorem{proposition}[theorem]{Proposition}
\newtheorem{corollary}[theorem]{Corollary}

\newtheorem{remark}[theorem]{Remark}
\newtheorem{lemma}[theorem]{Lemma}

\newcommand{\be}{\begin{equation}}
\newcommand{\ee}{\end{equation}}
\newcommand{\bea}{\begin{eqnarray}}
\newcommand{\eea}{\end{eqnarray}}
\newcommand{\ba}{\begin{array}}
	\newcommand{\ea}{\end{array}}
\newcommand{\bean}{\begin{eqnarray*}}
	\newcommand{\eean}{\end{eqnarray*}}
\newcommand{\al}{\alpha}

\newcommand{\La}{\Lambda}

\newcommand{\pa}{\partial}
\newcommand{\De}{\Delta}
\newcommand{\p}{\partial}

\newcommand{\E}{\mathcal{E}}

\begin{document}

\title{Lax formulation of 3--component KP hierarchy by Shiota construction}
\author{Tongtong Cui$^1$, Jinbiao Wang$^1$, Wenqi Cao$^1$, Jipeng Cheng $^{1,2*}$}
\dedicatory { $^{1}$School of Mathematics, China University of
Mining and Technology, Xuzhou, Jiangsu 221116,  China\\
$^{2}$ Jiangsu Center for Applied Mathematics (CUMT), \ Xuzhou, Jiangsu 221116, China}
\thanks{*Corresponding author. Email: chengjp@cumt.edu.cn, chengjipeng1983@163.com.}
\begin{abstract}
It is quite basic in integrable systems to deriving Lax equations from bilinear equations. For multi--component KP theory, corresponding Lax structures are mainly constructed by matrix pseudo--differential operators for fixed discrete variables, or by matrix difference operators for even--component cases. Here we use Shiota method to construct Lax structure of 3--component KP hierarchy and its reduction by introducing two shift operators $\Lambda_1$ and $\Lambda_2$, where relations among different discrete variables can be easily found. We believe the results here are quite typical for general multi--component KP theory, which may be helpful for general cases.

\end{abstract}

\maketitle

\section{Introduction}
KP theory \cite{Date1983,Dickey2003,Jimbo1983,Harnad2021,Ohta1988,vanMoerbeke1994,Kac2003,Miwa2000} has been playing an important role in mathematical physics and integrable systems. Just as we know, many research objects (e.g. Gromov--Witten invariants, Hurwitz number, Matrix models) \cite{Jimbo1983,Harnad2021,Ohta1988,vanMoerbeke1994} in mathematical physics can be viewed as KP tau functions. On the other hand, the usual integrable system equations, such as KdV, NLS and Davey--Stewartson equations, can be found in KP theory, which is just the universality of KP theory \cite{Dickey2003,Ohta1988,Kac2003}. Since KP hierarchy is too big to be used freely, different kinds of reductions are introduced to derive explicit differential equations. Among them, $n$--reduction of KP hierarchy, also called $n$--Gelfand--Dickey hierarchy\cite{Dickey2003}, is a quite typical one. Other famous reduction of KP hierarchy is the constrained KP hierarchy\cite{ChengY1992,Willox1999,ChengY1994,Oevel1993}, which is further generalization of $n$--Gelfand--Dickey hierarchy. Here in this paper, we are interested in multi--component KP hierarchy and its reduction \cite{Kac2020,Kac1998,Kac2003,Kac2023,Willox1999,Takasaki1984,DJKM1981}. For this, let us review their fermionic constructions. One can refer to \cite{DJKM1981,DJKM1982,Kac1998,Kac2003,Kac2023} for more details.

Recall that fermionic $n$--reduction of KP hierarchy \cite{Kac2023,Kac2003,DJKM1982} is defined by
\begin{align}
\Omega_{(nl)}^{\pm}(\tau\otimes\tau)=\sum_{p\in\mathbb{Z}+1/2}\psi_{-p}^{\pm}
\tau\otimes\psi^{\mp}_{p+nl}\tau=0,\quad \tau\in\mathcal{F},\quad l\in \mathbb{Z}_{\geq 0},
\end{align}
where $\psi_p^{\pm}$ is the charged free fermion satisfying
\begin{align*}
\psi_p^{\lambda}\psi_q^{\mu}+\psi_j^{\mu}\psi_i^{\lambda}=\delta_{\lambda,-\mu}\delta_{i,-j},\quad p,q\in \mathbb{Z}+1/2\ {\rm and}\ \lambda,\mu= +\ {\rm or}\ -,
\end{align*}
and fermionic Fock space $\mathcal{F}=\mathcal{A}|0\rangle$ with
$\mathcal{A}$ being the Clifford algebra generated by $\mathbf{1}$ and charged free fermions $\psi_p^{\pm}\left(p\in \mathbb{Z}+1/2\right)$, and
vacuum vector $|0\rangle$ defined by $\psi_p^{\pm}|0\rangle=0$ $(p>0)$.
Let $\underline{n}=\{n_1,n_2,\cdots,n_s\}$ satisfying $0\leq n_1\leq n_2\leq\cdots\leq n_s$ and $\sum_{i=1}^sn_i=n$ be a partition of $n$, and relabel charged free fermions in the following way
\begin{align*}
\psi^{\pm(i)}_{\pm\left(k+\frac{1}{2}\right)+n_il}=\psi^{\pm}_{\mp\left(n_1+n_2+\cdots+n_i-k-\frac{1}{2}\right)+n(l\pm1)},\quad 0\leq k\leq n_i-1,\quad l\in \mathbb{Z}.
\end{align*}
If define $\Omega_{(n_il)}^{\pm}=\sum_{p\in\mathbb{Z}+1/2}\psi_{-p}^{\pm(i)}
\otimes\psi^{\mp(i)}_{p+n_il}$, then fermionic $n$--reduction of KP hierarchy will become
\begin{align*}
\sum_{i=1}^s\Omega_{n_il}(\tau\otimes\tau)=0,\quad l\in \mathbb{Z}_{\geq 0}.
\end{align*}

By $n$--component boson--fermion correspondence, we can finally obtain \cite{DJKM1982,Kac2023,Kac2003}
\begin{align}
{\rm Res}_z\sum_{i=1}^{s}(-1)^{|\overrightarrow{m}+\overrightarrow{m}'|_{i-1}}z^{n_il\pm m_i\mp m'_i-2+2\delta_{is}}e^{\pm \xi(t^{(i)}-t'^{(i)},z)}\tau_{\overrightarrow{m}-\overrightarrow{e}_i+\overrightarrow{e}_s}(t\mp [z^{-1}]_i)
\tau_{\overrightarrow{m}'+\overrightarrow{e}_i-\overrightarrow{e}_s}(t'\pm [z^{-1}]_i)=0,\label{s-component-nKP}
\end{align}
where $t\mp [z^{-1}]_i=(t^{(1)},t^{(2)},\cdots,t^{(i-1)},t^{(i)}\mp [z^{-1}],t^{(i+1)},\cdots,t^{(s)})$, $t^{(k)}=(t^{(k)}_1,t^{(k)}_2,\cdots)$, $[z^{-1}]=(z^{-1},z^{-2}/2,z^{-3}/3,\cdots)$, $\overrightarrow{m},\overrightarrow{m}'\in\mathbb{Z}^s$ and $|\overrightarrow{m}|=|\overrightarrow{m}'|=0$, and $\overrightarrow{e}_i\in\mathbb{Z}^s$ are the standard basis vectors of $\mathbb{Z}^s$. Here in \eqref{s-component-nKP}, only $l=0$ and $l=1$ are independent, where \eqref{s-component-nKP} with $l=0$ is called $s$--component KP hierarchy ($s$--KP for short\cite{Kac2023,DJKM1981,Kac2003}), while \eqref{s-component-nKP} with $l=1$ is used to describe the constraints on $s$--KP hierarchy\cite{Kac2020,Kac2023,Kac2003}. The whole system of \eqref{s-component-nKP} (or only $l=0$ and $l=1$) is the full description for the $\underline{n}$--reduction of $s$--component KP hierarchy, which is called $\underline{n}$--KdV hierarchy\cite{Kac2020,Kac2023,Kac2003}. In particular,
\begin{itemize}
  \item $s=1$, \eqref{s-component-nKP} is $n$--Gelfand--Dickey hierarchy\cite{Dickey2003}
  \item $s=2$, \eqref{s-component-nKP} is $[n_1,n_2]$\ --\ bigraded Toda hierarchy\cite{Carlet2006}
\end{itemize}

In the investigation of $s$--KP or $\underline{n}$--KdV hierarchy, Lax description is quite important. Just as we can know, Lax equation \cite{Babelon2003} is one of important manifestations in integrable systems. However, there is no unified way to derive Lax equations from bilinear equations, which is still an open problem\cite{Kac1989,Kac1996}.
To our best knowledge, Lax structure of $s$--KP or $\underline{n}$--KdV is usually expressed in terms of matrix pseudo--differential operators of size $s\times s$ for fixed $\overrightarrow{m}\in\mathbb{Z}^s$\cite{Kac2020,Kac1998,Kac2003,Kac2023,DJKM1981,Takasaki1984}. Since the whole $s$--KP or $\underline{n}$--KdV depends on $\overrightarrow{m}\in\mathbb{Z}^s$, additional conditions are needed to relate the Lax structures of different $\overrightarrow{m}\in\mathbb{Z}^s$. And for the case $s$ even, $s$--KP or $\underline{n}$--KdV can also be formulated by matrix--difference operators of size $s\times s$\cite{Takasaki1984}, where the information $\overrightarrow{m}\in\mathbb{Z}^s$ is contained. Besides above Lax formulations of matrix operators, there is another formulation using scalar operators involving several differential operators $\pa_i$ or shift operators $\Lambda_i$, which is called Shiota construction\cite{ChengJP2021,Shiota1989}.
In our opinion, it is usually quite technical for matrix Lax formulation, especially in changing scalar bilinear equation into matrix forms. But in Shiota construction, one usually directly deals with the scalar bilinear equations by introducing several differential operators $\pa_i$ or shift operators $\Lambda_i$\cite{ChengJP2021,Shiota1989}. Therefore, we believe Shiota construction is comparatively direct compared with matrix formulation in dealing with Lax formulation of multi--component KP theory.

Here we will take 3--KP and its reduction $[n_1,n_2,n_3]$--KdV as examples to illustrate Shiota construction of Lax structures.
Notice that $[n_1,n_2,n_3]$--KdV hierarchy is given as follows,
  \begin{align}\label{3KPtaubilinear}
   &\oint_{C_R}\frac{dz}{2\pi i}z^{n_1l+m_1-m'_1}e^{\xi(t^{(1)}-t'^{(1)},z)}
   \widetilde{\tau}_{m_1,m_2}(t-[z^{-1}]_1)\widetilde{\tau}_{m'_1,m'_2}(t'+[z^{-1}]_1)\nonumber\\
   =&\oint_{C_r}\frac{dz}{2\pi i}z^{-n_2l+m_2-m'_2}e^{\xi(t^{(2)}-t'^{(2)},z^{-1})}
   \widetilde{\tau}_{m_1+1,m_2+1}(t-[z]_2)\widetilde{\tau}_{m'_1-1,m'_2-1}(t'+[z]_2)\nonumber\\
   &+\oint_{C_r}\frac{dz}{2\pi i}(-1)^{m_2+m'_2}z^{-n_3l+m_1-m_2-m'_1+m'_2}e^{\xi(t^{(3)}-t'^{(3)},z^{-1})}
   \widetilde{\tau}_{m_1+1,m_2}(t-[z]_3)\widetilde{\tau}_{m'_1-1,m'_2}(t'+[z]_3),
  \end{align}
where $C_R$ means the anticlockwise circle $|z|=R$ for sufficient large $R$, while $C_r$ is the anticlockwise circle $|z|=r$ with sufficient small $r$, and the tau function $\widetilde{\tau}_{m_1,m_2}$ is defined by
$$\widetilde{\tau}_{m_1,m_2}=(-1)^{\frac{m_1(m_1-1)}{2}+m_2}\tau_{m_1,m_2,-m_1-m_2}.$$
Here $\tau_{m_1,m_2,-m_1-m_2}$ is just the tau function in \eqref{s-component-nKP} for $s=3$.
The motivations of this paper are given as follows.
\begin{itemize}
  \item Try to find a unified way to investigate Lax structure of multi--component KP theory. Just as we stated before, there are Lax formulations of matrix difference operators in even--component case, therefore 3--component is quite typical, for which there are no difference operators to relate different discrete variables.

  \item Try to understand integrable systems involves two discrete variables. Notice that two shift operators are used in fractional Volterra hierarchy\cite{LiuSQ2018,LiuSQ2021}. We believe 3--component KP theory here may cover most these kinds of integrable systems.
\end{itemize}

This paper is organized in the way below. In Section 2, 3--KP hierarchy is expressed by wave operators involving two shift operators $\Lambda_1$ and $\Lambda_2$. Then in Section 3, relations of $\Lambda_1$ and $\Lambda_2$ are investigated. Next we give the Lax formulations of 3--KP hierarchy in Section 4. After that in Section 5, the Lax operator of $[n_1,n_2,n_3]$--KdV hierarchy is investigated. Finally, some conclusions and discussions are given in Section 6.

\section{3--KP Hierarchy by Wave Operators}

In this section, we will express 3--KP hierarchy \eqref{3KPtaubilinear} in terms of wave operators. Firstly some symbols are given. Then wave operators of 3--KP hierarchy are introduced. After that, starting from 3--KP bilinear equation \eqref{3KPtaubilinear} for $l=0$, relations between wave operators and evolution equations of wave operators are derived.

\subsection{Formal operators of $\Lambda_1$ and $\Lambda_2$}
Firstly let us introduce the following formal operator
 $$A=\sum_{j_1,j_2\in\mathbb{Z}}a_{j_1,j_2}(\bm{m})\Lambda_1^{j_1}\Lambda_2^{j_2},$$
where $\Lambda_1$ and $\Lambda_2$ are two shift operators defined by $\Lambda_1(f(\bm{m}))=f(\bm{m}+\bm{e}_1)$ and $\Lambda_2(f(\bm{m}))=f(\bm{m}+\bm{e}_2)$ with $\bm{m}=(m_1,m_2),\ \bm{e}_1=(1,0)$ and $\bm{e}_2=(0,1)$.
For another formal operator $B=\sum_{l_1,l_2\in\mathbb{Z}}b_{l_1,l_2}(\bm{m})\Lambda_1^{l_1}\Lambda_2^{l_2}$, we denote $AB$ or $A\cdot B$ to be the operator multiplication of $A$ and $B$, while $A(B)$ means that operator $A$ acts on coefficients of $B$. Then denote the following symbols for above operator $A$,
$$A^*=\sum_{j_1,j_2\in\mathbb{Z}}\Lambda_1^{-j_1}\Lambda_2^{-j_2}a_{j_1,j_2}(\bm{m}),\quad A_{i,P}=\sum_{j_i\in P,j_l\in\mathbb{Z}}a_{j_1,j_2}(\bm{m})\Lambda_1^{j_1}\Lambda_2^{j_2},\quad A_{i,[k]}=\sum_{j_i=k,j_l\in\mathbb{Z}}a_{j_1,j_2}(\bm{m})\Lambda_1^{j_1}\Lambda_2^{j_2},$$
where $l\in \{1,2\}\setminus \{i\}$ and $P\in\{\geq k, \leq k,>k,<k\}$ with $k\in\mathbb{Z}$. We also need $$\Delta_i=\Lambda_i-1,\quad\Delta_i^*=\Lambda_i^{-1}-1,\quad \Delta_{12}=\Lambda_1\Lambda_2-1,\quad \Delta_{12}^{*}=\Lambda_1^{-1}\Lambda_2^{-1}-1,$$
and denote for $Q\in\{\Lambda_1,\Lambda_2,\Lambda_1\Lambda_2,
\Lambda_1^{-1},\Lambda_2^{-1},\Lambda_1^{-1}\Lambda_2^{-1}\}$ and $R\in\{\Delta_1,\Delta_2,\Delta_{12},\Delta_1^*,\Delta_2^*,\Delta_{12}^*\}$
$$(Q-1)^{-1}=\sum_{j=1}^\infty Q^{-j},\quad (R+1)^{-k}=\sum_{j=0}^\infty \binom{-k}{j}R^{-k-j}.$$
Then we can rewrite the operator $A=\sum_{j_1,j_2\in\mathbb{Z}}a_{j_1,j_2}(\bm{m})\Lambda_1^{j_1}\Lambda_2^{j_2}$ in terms of $(\Delta_1,\Delta_2)$ or $(\Delta_1^*,\Delta_2^*)$. Now similar to $A_{i,P}$ with $P\in\{\geq k, \leq k,>k,<k\}$, we can also define $A_{\Delta_i,P}$ (or $A_{\Delta_i^*,P}$) to be the part of $A$ satisfying property $P$ with respect to operator $\Delta_i$ (or $\Delta_i^*$).

\begin{lemma}\cite{Adler1999}\label{ABlemma}
Let $A(m,\Lambda)=\sum_{j}a_j(m)\Lambda^j,\ B(m,\Lambda)=\sum_{j}b_j(m)\Lambda^j$ be two operators with shift operator $\Lambda$ defined by $\Lambda(f(m))=f(m+1)$, then
\begin{align*}
A(m,\Lambda)\cdot B(m,\Lambda)^*=\sum_{j\in\mathbb{Z}}{\rm Res}_zz^{-1}\left(A(m,\Lambda)(z^{\pm m})\cdot
B(m+j,\Lambda) (z^{\mp m\mp j})\right)\Lambda^j,
\end{align*}
\end{lemma}
By similar methods in Lemma \ref{ABlemma}, we can get another lemma.
\begin{lemma}\label{ABlemma2}
Let $A(\bm{m},\Lambda_i)=\sum_{k}a_k(\bm{m})\Lambda_i^k,\ B(\bm{m},\Lambda_i)=\sum_{k}b_k(\bm{m})\Lambda_i^k$ ($i=1,2$) be two operators, then
\begin{align*}
&A(\bm{m},\Lambda_1)B^*(\bm{m}+\textbf{j},\Lambda_1)=\sum_{l\in\mathbb{Z}}
{\rm Res}_zz^{-1}A(\bm{m},\Lambda_1)(z^{m_1-m_2})B(\bm{m}+\textbf{j}+l\bm{e}_1,\Lambda_1)(z^{-m_1+m_2-l})\Lambda_1^l,\\
&A(\bm{m},\Lambda_2)B^*(\bm{m}+\textbf{j},\Lambda_2)=\sum_{l\in\mathbb{Z}}
{\rm Res}_zz^{-1}A(\bm{m},\Lambda_2)(z^{m_1-m_2})B(\bm{m}+\textbf{j}+l\bm{e}_2,\Lambda_2)(z^{-m_1+m_2+l})\Lambda_2^l.
\end{align*}
where $\bm{j}=(j_1,j_2)$.
\end{lemma}
\begin{lemma}\cite{Guan2024}\label{lemma:Alade}
Given $A=\sum_{j}a_j(m)\Lambda^j$, $\Delta=\Lambda-1$ and $\Delta^*=\Lambda^{-1}-1$, we have
\begin{align*}
A_{\Lambda,\geq 0}=A_{\Delta,\geq 0},\quad A_{\Lambda,\leq 0}=A_{\Delta^*,\geq 0},
\end{align*}
where $Q\in \{\Lambda,\Delta,\Delta^*\}$, $P\in \{\geq k,>k,\leq k,<k\}$, and $A_{Q,P}$ means the part of $A$ satisfying property $P$ with respect to $Q$. Further
\begin{align*}
A_{\Delta,\geq 1}=A_{\Lambda,\geq 0}-A_{\Lambda,\geq 0}(1),\quad A_{\Delta^*,\geq 1}=A_{\Lambda,< 0}-A_{\Lambda,< 0}(1).
\end{align*}
\end{lemma}
\subsection{Wave operators}

Firstly introduce wave functions $\Psi_i$ and adjoint wave functions $\widetilde{\Psi}_i$ by the way below,
\begin{align*}
&\Psi_{1}(\bm{m},t,z)=z^{m_1}e^{\xi(t^{(1)},z)}
\frac{\widetilde{\tau}_{\bm{m}}(t-[z^{-1}]_1)}{\widetilde{\tau}_{\bm{m}}(t)},\quad\quad\quad\quad\quad\ \ \
\widetilde{\Psi}_{1}(\bm{m},t,z)=z^{-m_1+1}e^{-\xi(t^{(1)},z)}
\frac{\widetilde{\tau}_{\bm{m}}(t+[z^{-1}]_1)}{\widetilde{\tau}_{\bm{m}}(t)},\\
&\Psi_{2}(\bm{m},t,z)=z^{m_2}e^{\xi(t^{(2)},z^{-1})}
\frac{\widetilde{\tau}_{\bm{m}+\bm{e}}(t-[z]_2)}
{\widetilde{\tau}_{\bm{m}}(t)},\quad\quad\quad\quad\quad\
\widetilde{\Psi}_{2}(\bm{m},t,z)=z^{-m_2+1}e^{-\xi(t^{(2)},z^{-1})}
\frac{\widetilde{\tau}_{\bm{m}-\bm{e}}(t+[z]_2)}
{\widetilde{\tau}_{\bm{m}}(t)},\\
&\Psi_{3}(\bm{m},t,z)=(-1)^{m_2}z^{m_1-m_2}e^{\xi(t^{(3)},z^{-1})}
\frac{\widetilde{\tau}_{\bm{m}+\bm{e}_1}(t-[z]_3)}{\widetilde{\tau}_{\bm{m}}(t)},\quad
\widetilde{\Psi}_{3}(\bm{m},t,z)=(-1)^{m_2}z^{-m_1+m_2+1}e^{-\xi(t^{(3)},z^{-1})}
\frac{\widetilde{\tau}_{\bm{m}-\bm{e}_1}(t+[z]_3)}{\widetilde{\tau}_{\bm{m}}(t)},
\end{align*}
then the bilinear equation \eqref{3KPtaubilinear} can be rewritten into
\begin{align}\label{3KPwavebilinear}
&\oint_{C_R}\frac{dz}{2\pi iz}z^{n_1l}\Psi_{1}(\bm{m},t,z)\widetilde{\Psi}_{1}(\bm{m}',t',z)\nonumber\\
=&\oint_{C_r}\frac{dz}{2\pi iz}\left(z^{-n_2l}\Psi_{2}(\bm{m},t,z)\widetilde{\Psi}_{2}(\bm{m}',t',z)
+z^{-n_3l}\Psi_{3}(\bm{m},t,z)\widetilde{\Psi}_{3}(\bm{m}',t',z)\right),\quad l\geq 0.
\end{align}
After preparation above, let us introduce operators $W_i,\ \widetilde{W}_i$ and $S_i,\ \widetilde{S}_i\ (i=1,2,3)$ as follows,
\begin{align*}
&W_{1}(\bm{m},t,\Lambda_1)= S_{1}(\bm{m},t,\Lambda_1)e^{\xi(t^{(1)},\Lambda_1)},\quad\ \
\widetilde{W}_{1}(\bm{m},t,\Lambda_1)= \widetilde{S}_{1}(\bm{m},t,\Lambda_1)e^{-\xi(t^{(1)},\Lambda_1^{-1})}\Lambda_1^{-1},\\
&W_{2}(\bm{m},t,\Lambda_2)= S_{2}(\bm{m},t,\Lambda_2)e^{\xi(t^{(2)},\Lambda_2^{-1})},\quad\
\widetilde{W}_{2}(\bm{m},t,\Lambda_2)=\widetilde{S}_{2}(\bm{m},t,\Lambda_2)
e^{-\xi(t^{(2)},\Lambda_2)}\Lambda_2^{-1},\\
&W_{3}(\bm{m},t,\Lambda_1)= S_{3}(\bm{m},t,\Lambda_1)e^{\xi(t^{(3)},\Lambda_1^{-1})},\quad\ \
\widetilde{W}_{3}(\bm{m},t,\Lambda_1)= \widetilde{S}_{3}(\bm{m},t,\Lambda_1)e^{-\xi(t^{(3)},\Lambda_1)}\Lambda_1^{-1},
\end{align*}
and
\begin{align*}
&S_{1}(\bm{m},t,\Lambda_1)=1+\sum_{k=1}^{+\infty}a^{(1)}_k(t)\Lambda_1^{-k},\quad\quad\quad\quad\quad\quad\ \
\widetilde{S}_{1}(\bm{m},t,\Lambda_1)=1+\sum_{k=1}^{+\infty}\widetilde{a}^{(1)}_k(t)\Lambda_1^{k},\\
&S_{2}(\bm{m},t,\Lambda_2)=\frac{\widetilde{\tau}_{\bm{m}+\bm{e}}(t)}
{\widetilde{\tau}_{\bm{m}}(t)}+\sum_{k=1}^{+\infty}a^{(2)}_k(t)\Lambda_2^{k},\quad\quad\quad\quad\
\widetilde{S}_{2}(\bm{m},t,\Lambda_2)=\frac{\widetilde{\tau}_{\bm{m}-\bm{e}}(t)}
{\widetilde{\tau}_{\bm{m}}(t)}+\sum_{k=1}^{+\infty}\widetilde{a}^{(2)}_k(t)\Lambda_2^{-k},\\
&S_{3}(\bm{m},t,\Lambda_1)=(-1)^{m_2}\frac{\widetilde{\tau}_{\bm{m}+\bm{e}_1}(t)}
{\widetilde{\tau}_{\bm{m}}(t)}+\sum_{k=1}^{+\infty}a^{(3)}_k(t)\Lambda_1^{k},\quad
\widetilde{S}_{3}(\bm{m},t,\Lambda_1)=(-1)^{m_2}\frac{\widetilde{\tau}_{\bm{m}-\bm{e}_1}(t)}
{\widetilde{\tau}_{\bm{m}}(t)}+\sum_{k=1}^{+\infty}\widetilde{a}^{(3)}_k(t)\Lambda_1^{-k}.
\end{align*}
satisfying
\begin{align*}
&\Psi_{1}(\bm{m},t,z)=W_{1}(\bm{m},t,\Lambda_1)(z^{m_1}),\quad\quad\ \
\widetilde{\Psi}_{1}(\bm{m},t,z)=\widetilde{W}_{1}(\bm{m},t,\Lambda_1)(z^{-m_1});\nonumber\\
&\Psi_{2}(\bm{m},t,z)=W_{2}(\bm{m},t,\Lambda_2)(z^{m_2}),\quad\quad\ \
\widetilde{\Psi}_{2}(\bm{m},t,z)=\widetilde{W}_{2}(\bm{m},t,\Lambda_2)(z^{-m_2});\\
&\Psi_{3}(\bm{m},t,z)=W_{3}(\bm{m},t,\Lambda_1)(z^{m_1-m_2})=W_{3}(\bm{m},t,\Lambda_2^{-1})(z^{m_1-m_2}),\\
&\widetilde{\Psi}_{3}(\bm{m},t,z)=\widetilde{W}_{3}(\bm{m},t,\Lambda_1)(z^{-m_1+m_2})=\widetilde{W}_{3}(\bm{m},t,\Lambda_2^{-1})(z^{-m_1+m_2})\nonumber.
\end{align*}
Here for $Q\in\{S_3,W_3,\widetilde{S}_3,\widetilde{W}_3\}$, $Q(\bm{m},t,\Lambda_2^{-1})$ is obtained by replacing $\Lambda_1$ in $Q(\bm{m},t,\Lambda_1)$ with $\Lambda_2^{-1}$ without changing places of $\Lambda_1$ in above relations of $Q(\bm{m},t,\Lambda_1)$. In what follows, we also use $R(\bm{m},\Lambda_i^{\pm 1})$ or $R(\Lambda_i^{\pm 1})$ for brevity to instead of $R(\bm{m},t,\Lambda_i^{\pm 1})$.

\subsection{Relations between $S_i$ and $\widetilde{S}_i$}

If set $\bm{m}'=\bm{m}+\bm{j}$ with $\bm{j}=(j_1,j_2)$ in (\ref{3KPwavebilinear}), we can get
\begin{align}
&\sum_{j_1,j_2\in\mathbb{Z}}\oint_{C_R}\frac{dz}{2\pi iz}z^l\Psi_1(\bm{m},t,z)\widetilde{\Psi}_1(\bm{m}+\bm{j},t',z)\Lambda_1^{j_1}\Lambda_2^{j_2}\nonumber\\
=&\sum_{j_1,j_2\in\mathbb{Z}}\oint_{C_r}\frac{dz}{2\pi iz}z^l\left(\Psi_2(\bm{m},t,z)\widetilde{\Psi}_2(\bm{m}+\bm{j},t',z)\Lambda_1^{j_1}\Lambda_2^{j_2}
+\Psi_3(\bm{m},t,z)\widetilde{\Psi}_3(\bm{m}+\bm{j},t',z)\Lambda_1^{j_1}\Lambda_2^{j_2}\right).\label{3kpbilinear-wave-sum}
\end{align}
Further by Lemma \ref{ABlemma} and Lemma \ref{ABlemma2}, we can obtain the following proposition.
\begin{proposition}\label{bilinear-wave-operator}
Wave operators $S_i$ and $\widetilde{S}_i$ satisfy ($l\geq 0$)
\begin{align*}
&\sum_{j_2\in\mathbb{Z}}S_1(\bm{m},t,\Lambda_1)\Lambda_1^{n_1l+1}e^{\xi(t^{(1)}-t'^{(1)},\Lambda_1)}
\widetilde{S}^*_1(\bm{m}+j_2\bm{e}_2,t',\Lambda_1)\Lambda_2^{j_2}\\
-&\sum_{j_1\in\mathbb{Z}}S_2(\bm{m},t,\Lambda_2)\Lambda_2^{-n_2l+1}e^{\xi(t^{(2)}-t'^{(2)},\Lambda_2^{-1})}
\widetilde{S}^*_2(\bm{m}+j_1\bm{e}_1,t',\Lambda_2)\Lambda_1^{j_1}\\
=&\sum_{j_1\in\mathbb{Z}}S_3(\bm{m},t,\Lambda_1)\Lambda_1^{-n_3l+1}e^{\xi(t^{(3)}-t'^{(3)},\Lambda_1^{-1})}
\widetilde{S}^*_3(\bm{m}+j_1\bm{e},t',\Lambda_1)\Lambda_1^{j_1}\Lambda_2^{j_1}\\
=&\sum_{j_1\in\mathbb{Z}}S_3(\bm{m},t,\Lambda_2^{-1})\Lambda_2^{n_3l-1}e^{\xi(t^{(3)}-t'^{(3)},\Lambda_2)}
\widetilde{S}^*_3(\bm{m}+j_1\bm{e},t',\Lambda_2^{-1})\Lambda_1^{j_1}\Lambda_2^{j_1}.
\end{align*}
\end{proposition}
\begin{proposition}\label{relations between S_i}
The relations between $S_i$ and $\widetilde{S}_i$ $(i=1,2,3)$ are given by
\begin{align*}
&S_1(\bm{m},\Lambda_1)\Lambda_1\widetilde{S}_1^*(\bm{m},\Lambda_1)=\Lambda_1,\quad
S_2(\bm{m},\Lambda_2)\Lambda_2\widetilde{S}_2^*(\bm{m}+\bm{e}_1,\Lambda_2)=(\Delta_2^*)^{-1},\\
&S_3(\bm{m},\Lambda_1)\Lambda_1\widetilde{S}_3^*(\bm{m},\Lambda_1)=\Lambda_1,\quad
S_3(\bm{m},\Lambda_2^{-1})\Lambda_2^{-1}
\widetilde{S}_3^*(\bm{m}+\bm{e},\Lambda_2^{-1})=\Delta_2^{-1}.
\end{align*}
\end{proposition}
\begin{proof}
Let $t'=t,\ l=0$ in the relation of Proposition \ref{bilinear-wave-operator}, then
\begin{align}
&\sum_{j\in\mathbb{Z}}S_1(\bm{m},\Lambda_1)\Lambda_1\widetilde{S}_1^*(\bm{m}+j\bm{e}_2,\Lambda_1)\Lambda_2^j
-\sum_{j\in\mathbb{Z}}S_2(\bm{m},\Lambda_2)\Lambda_2\widetilde{S}_2^*(\bm{m}+j\bm{e}_1,\Lambda_2)\Lambda_1^j\nonumber\\
=&\sum_{j\in\mathbb{Z}}S_3(\bm{m},\Lambda_1)\Lambda_1\widetilde{S}_3^*(\bm{m}+j\bm{e},\Lambda_1)\Lambda_1^j\Lambda_2^j
=\sum_{j\in\mathbb{Z}}S_3(\bm{m},\Lambda_2^{-1})\Lambda_2^{-1}
\widetilde{S}_3^*(\bm{m}+j\bm{e},\Lambda_2^{-1})\Lambda_1^j\Lambda_2^j.\label{ssrelation}
\end{align}
Firstly let us compare the coefficients of $\Lambda_2^0$ in \eqref{ssrelation}. Notice that $\left(S_2(\bm{m},\Lambda_2)\Lambda_2
\widetilde{S}_2^*(\bm{m}+j\bm{e}_1,\Lambda_2)\right)_{2,[0]}=0$, therefore we can obtain $S_1(\bm{m},\Lambda_1)\Lambda_1\widetilde{S}_1^*(\bm{m},\Lambda_1)
=S_3(\bm{m},\Lambda_1)\Lambda_1\widetilde{S}_3^*(\bm{m},\Lambda_1)$. Further $\Lambda_1$ is the highest order term of $S_1(\bm{m},\Lambda_1)\Lambda_1\widetilde{S}_1^*(\bm{m},\Lambda_1)$ with respect to $\Lambda_1$. Thus we can at last obtain that
\begin{align*}
S_1(\bm{m},\Lambda_1)\Lambda_1\widetilde{S}_1^*(\bm{m},\Lambda_1)
=S_3(\bm{m},\Lambda_1)\Lambda_1\widetilde{S}_3^*(\bm{m},\Lambda_1)=\Lambda_1.
\end{align*}
Next if compare the coefficients of $\Lambda_1$ in \eqref{ssrelation}, we can obtain
\begin{align*}
\sum_{j\in\mathbb{Z}}\Lambda_2^j-S_2(\bm{m},\Lambda_2)\Lambda_2\widetilde{S}_2^*(\bm{m}+\bm{e}_1,\Lambda_2)
=S_3(\bm{m},\Lambda_2^{-1})\Lambda_2^{-1}
\widetilde{S}_3^*(\bm{m}+\bm{e},\Lambda_2^{-1})\Lambda_2.
\end{align*}
Notice that $S_2(\bm{m},\Lambda_2)\Lambda_2\widetilde{S}_2^*(\bm{m}+\bm{e}_1,\Lambda_2)$ has positive $\Lambda_2$--order, while $S_3(\bm{m},\Lambda_2^{-1})\Lambda_2^{-1}
\widetilde{S}_3^*(\bm{m}+\bm{e},\Lambda_2^{-1})\Lambda_2$ has non--positive $\Lambda_2$--order. Based upon these two facts, one can easily obtain relations between $S_i$ and $\widetilde{S}_i$ for $i=2,3$.
\end{proof}
\subsection{Evolution equations of wave operators}
\begin{proposition}\label{prop:evolutionS}
Evolution equation of wave operators with respect to $t_k^{(i)}(i=1,2,3)$ are given as follows, which are called $t_k^{(i)}-$ flows.
\begin{itemize}
  \item $t_k^{(1)}$--flows:
  \begin{align*}
  &\partial_{t_k^{(1)}}S_1(\bm{m},\Lambda_1)
=-\left(S_1(\bm{m},\Lambda_1)\Lambda_1^{k}S_1^{-1}(\bm{m},\Lambda_1)\right)_{1,<0}S_1(\bm{m},\Lambda_1),\\
&\partial_{t_k^{(1)}}S_3(\bm{m},\Lambda_1)
=\left(S_1(\bm{m},\Lambda_1)\Lambda_1^{k}S_1^{-1}(\bm{m},\Lambda_1)\right)_{1,\geq0}S_3(\bm{m},\Lambda_1),\\
&\partial_{t_k^{(1)}}S_2(\bm{m},\Lambda_2)
=\left(S_1(\bm{m},\Lambda_1)\Lambda_1^{k}(\Delta_2^*)^{-1}S_1^{-1}(\bm{m},\Lambda_1)\right)_{1,[0]}\Delta_2^*S_2(\bm{m},\Lambda_2),\\
&\partial_{t_k^{(1)}}S_3(\bm{m},\Lambda_2^{-1})
=\left(S_1(\bm{m},\Lambda_1)\Lambda_1^{k}\Lambda_2\Delta_2^{-1}
S_1^{-1}(\bm{m},\Lambda_1)\right)_{1,[0]}\Lambda_2^{-1}\Delta_2S_3(\bm{m},\Lambda_2^{-1}).
  \end{align*}
  \item $t_k^{(2)}$--flows:
  \begin{align*}
  &\partial_{t_k^{(2)}}S_1(\bm{m},\Lambda_1)
=\left(S_2(\bm{m},\Lambda_2)\Lambda_2^{-k}\Delta_1^{-1}
S_2^{-1}(\bm{m},\Lambda_2)(\Delta_2^*)^{-1}\right)_{2,[0]}S_1(\bm{m},\Lambda_1),\\
&\partial_{t_k^{(2)}}S_3(\bm{m},\Lambda_1)=-\left(S_2(\bm{m},\Lambda_2)\Lambda_2^{-k}\Lambda_1^{-1}\Delta_1^{*-1}
S_2^{-1}(\bm{m},\Lambda_2)(\Delta_2^*)^{-1}\right)_{2,[0]}S_3(\bm{m},\Lambda_1),\\
&\partial_{t_k^{(2)}}S_2(\bm{m},\Lambda_2)=-\left(S_2(\bm{m},\Lambda_2)\Lambda_2^{-k}
S_2^{-1}(\bm{m},\Lambda_2)\right)_{\Delta_2^*,\leq0}S_2(\bm{m},\Lambda_2),\\
&\partial_{t_k^{(2)}}S_3(\bm{m},\Lambda_2^{-1})=\left(S_2(\bm{m},\Lambda_2)\Lambda_2^{-k}
S_2^{-1}(\bm{m},\Lambda_2)\right)_{\Delta_2^*,\geq 1}S_3(\bm{m},\Lambda_2^{-1}).
  \end{align*}
  \item $t_k^{(3)}$--flows
  \begin{itemize}
    \item In terms of $\Lambda_1$--operator and $S_3(\bm{m},\Lambda_1)$
    \begin{align*}
    &\partial_{t_k^{(3)}}S_1(\bm{m},\Lambda_1)=\left(S_3(\bm{m},\Lambda_1)\Lambda_1^{-k}
S_3^{-1}(\bm{m},\Lambda_1)\right)_{1,<0}S_1(\bm{m},\Lambda_1),\\
&\partial_{t_k^{(3)}}S_2(\bm{m},\Lambda_2)=-\left(S_3(\bm{m},\Lambda_1)\Lambda_1^{-k}(\Delta_{12}^*)^{-1}
S_3^{-1}(\bm{m},\Lambda_1)\right)_{1,[0]}\Delta_2^*S_2(\bm{m},\Lambda_2),\\
&\partial_{t_k^{(3)}}S_3(\bm{m},\Lambda_1)=-\left(S_3(\bm{m},\Lambda_1)\Lambda_1^{-k}
S_3^{-1}(\bm{m},\Lambda_1)\right)_{1,\geq0}S_3(\bm{m},\Lambda_1).
    \end{align*}
    \item In terms of $\Lambda_2$--operator and $S_3(\bm{m},\Lambda_2^{-1})$
    \begin{align*}
&\partial_{t_k^{(3)}}S_1(\bm{m},\Lambda_1)=\left(S_3(\bm{m},\Lambda_2^{-1})\Lambda_2^{k+1}\Lambda_1\Delta_{12}^{-1}
S_3^{-1}(\bm{m},\Lambda_2^{-1})\Delta_2^{-1}\Lambda_2\right)_{2,[0]}S_1(\bm{m},\Lambda_1),\\
&\partial_{t_k^{(3)}}S_2(\bm{m},\Lambda_2)=\left(S_3(\bm{m},\Lambda_2^{-1})\Lambda_2^k
S_3^{-1}(\bm{m},\Lambda_2^{-1})\right)_{\Delta_2,\geq1}S_2(\bm{m},\Lambda_2),\\
&\partial_{t_k^{(3)}}S_3(\bm{m},\Lambda_2^{-1})=-\left(S_3(\bm{m},\Lambda_2^{-1})\Lambda_2^k
S_3^{-1}(\bm{m},\Lambda_2^{-1})\right)_{\Delta_2,\leq0}S_3(\bm{m},\Lambda_2^{-1}).
\end{align*}
  \end{itemize}
\end{itemize}

\end{proposition}
\begin{proof}
Firstly if apply $\partial_{t_k^{(1)}}$ to the relation in Proposition \ref{bilinear-wave-operator} and let $t'=t,\ l=0$, then we have
\begin{align}
&\sum_{j\in\mathbb{Z}}\left(\partial_{t_k^{(1)}}S_1(\bm{m},\Lambda_1)
+S_1(\bm{m},\Lambda_1)\Lambda_1^{k}\right)\cdot S_1^{-1}(\bm{m}+j\bm{e}_2,\Lambda_1)\Lambda_2^{j}\Lambda_1\nonumber\\
-&\sum_{j\in\mathbb{Z}}\partial_{t_k^{(1)}}S_2(\bm{m},\Lambda_2)\cdot
S_2^{-1}(\bm{m}+(j-1)\bm{e}_1,m_2,\Lambda_2)(\Delta_2^*)^{-1}\Lambda_1^j\nonumber\\
=&\sum_{j\in\mathbb{Z}}\partial_{t_k^{(1)}}S_3(\bm{m},\Lambda_1)\cdot
S_3^{-1}(\bm{m}+j\bm{e},\Lambda_1)\Lambda_1^{j+1}\Lambda_2^{j}\nonumber\\
=&\sum_{j\in\mathbb{Z}}\partial_{t_k^{(1)}}S_3(\bm{m},\Lambda_2^{-1})\cdot
S_3^{-1}(\bm{m}+(j-1)\bm{e},\Lambda_2^{-1})\Delta_2^{-1}\Lambda_1^{j}\Lambda_2^{j}.\label{patk1-3kpbilinear}
\end{align}
Notice that $\left(\partial_{t_k^{(1)}}S_2(\bm{m},\Lambda_2)
S_2^{-1}(\bm{m}+(j-1)\bm{e}_1,\Lambda_2)(\Delta_2^*)^{-1}\right)_{2,[0]}=0$, so by comparing coefficients of $\Lambda_2^0$ in \eqref{patk1-3kpbilinear}, we can get
\begin{align}
\left(\partial_{t_k^{(1)}}S_1(\bm{m},\Lambda_1)+S_1(\bm{m},\Lambda_1)\Lambda_1^{k}\right)
\cdot S_1^{-1}(\bm{m},\Lambda_1)
=\partial_{t_k^{(1)}}S_3(\bm{m},\Lambda_1)S_3^{-1}(\bm{m},\Lambda_1).
\label{patk1-3kpbilinear-lambda_20}
\end{align}
which implies the results for $\partial_{t_k^{(1)}}S_1(\bm{m},\Lambda_1)$ and $\partial_{t_k^{(1)}}S_3(\bm{m},\Lambda_1)$ by taking the terms in \eqref{patk1-3kpbilinear-lambda_20} with negative and non--negative $\Lambda_1$--orders respectively. Next if consider coefficients of $\Lambda_1$ in \eqref{patk1-3kpbilinear}, one can get
\begin{align}
&\sum_{j\in\mathbb{Z}}\left(S_1(\bm{m},\Lambda_1)\Lambda_1^{k}S_1^{-1}(\bm{m}+j\bm{e}_2,\Lambda_1)\right)_{1,[0]}\Lambda_2^j
-\partial_{t_k^{(1)}}S_2(\bm{m},\Lambda_2)\cdot S_2^{-1}(\bm{m},\Lambda_2)(\Delta_2^*)^{-1}\nonumber\\
&=\partial_{t_k^{(1)}}S_3(\bm{m},\Lambda_2^{-1})\cdot S_3^{-1}(\bm{m},\Lambda_2^{-1})\Delta_2^{-1}\Lambda_2,\label{patk1-3kpbilinear-lambda_11}
\end{align}
where we have used $\left(\partial_{t_k^{(1)}}S_1(\bm{m},\Lambda_1)\cdot S_1^{-1}(\bm{m}+j\bm{e}_2,\Lambda_1)\right)_{1,[0]}=0$. Then $\partial_{t_k^{(1)}}S_2(\bm{m},\Lambda_2)$ comes from the terms in \eqref{patk1-3kpbilinear-lambda_11} with positive $\Lambda_2$--orders, while the terms with non--positive $\Lambda_2$--orders give rise to $\partial_{t_k^{(1)}}S_3(\bm{m},\Lambda_2^{-1})$.

We can use similar method above to obtain $\partial_{t_k^{(2)}}S_1(\bm{m},\Lambda_1)$, $\partial_{t_k^{(2)}}S_3(\bm{m},\Lambda_1)$ and
\begin{align*}
&\partial_{t_k^{(2)}}S_2(\bm{m},\Lambda_2)=-\left(S_2(\bm{m},\Lambda_2)\Lambda_2^{-k}
S_2^{-1}(\bm{m},\Lambda_2)(\Delta_2^*)^{-1}\right)_{2,\geq1}\Delta_2^*S_2(\bm{m},\Lambda_2),\\
&\partial_{t_k^{(2)}}S_3(\bm{m},\Lambda_2^{-1})=-\left(S_2(\bm{m},\Lambda_2)\Lambda_2^{-k}
S_2^{-1}(\bm{m},\Lambda_2)(\Delta_2^*)^{-1}\right)_{2,\leq0}\Lambda_2^{-1}\Delta_2S_3(\bm{m},\Lambda_2^{-1}).
\end{align*}
Then final results of $\partial_{t_k^{(2)}}S_2(\bm{m},\Lambda_2)$ and $\partial_{t_k^{(2)}}S_3(\bm{m},\Lambda_2^{-1})$ can be obtained by Lemma \ref{lemma:Alade}. As for $t_k^{(3)}$--flows, they can be derived by similar way to Cases for $t_k^{(1)}$ and $t_k^{(2)}$--flows.
\end{proof}

\section{Relations of $\Lambda_1$ and $\Lambda_2$}
In this section, we will investigate relations of $\Lambda_1$ and $\Lambda_2$ in 3$-$KP hierarchy. Firstly, relations of $\Lambda_1^{\pm k}$ and $\Lambda_2^{\pm k}$ ($k>0$) on wave operators are derived from 3--KP bilinear equation. Then we restrict these relations to the case of $k=1$ and obtain one important operator $H$, which acts trivially on wave functions. After that, we introduce some operator spaces involving $\Lambda_1$ and $\Lambda_2$ and consider the corresponding decompositions. Finally we define four kinds of projections to relate $\Lambda_1$ with $\Lambda_2$ and some formulas are given.

\begin{proposition}\label{prop:lambda12relation}
Given $k>0$, the relations between $\Lambda_1$ and $\Lambda_2$ are given as follows.
\begin{itemize}
  \item $\Lambda_1$ acting on $S_2(\bm{m},\Lambda_2)$ and $S_3(\bm{m},\Lambda_2^{-1})$
  \begin{align*}
&\Lambda_1^k(S_2(\bm{m},\Lambda_2))=\left(\Lambda_1^kS_1(\bm{m},\Lambda_1)(\Delta_2^*)^{-1}
S_1^{-1}(\bm{m},\Lambda_1)\right)_{1,[0]}\cdot\Delta_2^*S_2(\bm{m},\Lambda_2),\\
&\Lambda_1^k(S_3(\bm{m},\Lambda_2^{-1}))=\left(\Lambda_1^kS_1(\bm{m},\Lambda_1)\Delta_2^{-1}
S_1^{-1}(\bm{m},\Lambda_1)\right)_{1,[0]}\cdot\Lambda_2^{-1}\Delta_2S_3(\bm{m},\Lambda_2^{-1})\Lambda_2^k,\\
&\Lambda_1^{-k}(S_2(\bm{m},\Lambda_2))=-\left(\Lambda_1^{-k}S_3(\bm{m},\Lambda_1)(\Delta_{12}^*)^{-1}
S_3^{-1}(\bm{m},\Lambda_1)\right)_{1,[0]}\cdot\Delta_2^*S_2(\bm{m},\Lambda_2),\\
&\La_1^{-k}(S_3(\bm{m},\La_2^{-1}))=(\La_1^{-k}S_3(\bm{m},\La_1)(\Delta_{12}^*)^{-1}S_3^{-1}(\bm{m},\La_1))_{1,[0]}\cdot
\Lambda_2^{-1}\Delta_2S_3(\bm{m},\La_2^{-1})\La_2^{-k}.
\end{align*}
  \item $\Lambda_2$ acting on $S_1(\bm{m},\Lambda_1)$ and $S_3(\bm{m},\Lambda_1)$
\begin{align*}
&\Lambda_2^k(S_1(\bm{m},\Lambda_1))=\left(\left(\Lambda_2^kS_3(\bm{m},\Lambda_2^{-1})\Delta_{12}^{-1}
S_3^{-1}(\bm{m},\Lambda_2^{-1})\Lambda_2\Delta_2^{-1}\right)_{2,[0]}+1\right)\cdot S_1(\bm{m},\Lambda_1),\\
&\Lambda_2^{k}(S_3(\bm{m},\Lambda_1))=\left(\left(\Lambda_2^kS_3(\bm{m},\Lambda_2^{-1})\Delta_{12}^{-1}
S_3^{-1}(\bm{m},\Lambda_2^{-1})\Lambda_2\Delta_2^{-1}\right)_{2,[0]}+1\right)\cdot S_3(\bm{m},\Lambda_1)\Lambda_1^{k},\\
&\Lambda_2^{-k}(S_1(\bm{m},\Lambda_1))=\left(\left(\Lambda_2^{-k}S_2(\bm{m},\Lambda_2)\Delta_1^{-1}
S_2^{-1}(\bm{m},\Lambda_2)(\Delta_2^*)^{-1}\right)_{2,[0]}+1\right)\cdot S_1(\bm{m},\Lambda_1),\\
&\Lambda_2^{-k}(S_3(\bm{m},\Lambda_1))=-\left(\Lambda_2^{-k}S_2(\bm{m},\Lambda_2)(\Delta_1^*)^{-1}
S_2^{-1}(\bm{m},\Lambda_2)(\Delta_2^*)^{-1}\right)_{2,[0]}\cdot S_3(\bm{m},\Lambda_1)\Lambda_1^{-k}.
\end{align*}
\end{itemize}
\end{proposition}
\begin{proof}
Firstly by setting $\bm{m}\rightarrow \bm{m}+l\bm{e}_1$ in \eqref{3KPwavebilinear}, we can get the similar relation \eqref{3kpbilinear-wave-sum} with $\bm{m}\rightarrow \bm{m}+l\bm{e}_1$ in $\Psi_i$. Further by using Lemma \ref{ABlemma} and Lemma \ref{ABlemma2}, we have
\begin{align*}
&\sum_{j\in\mathbb{Z}}S_1(\bm{m}+l\bm{e}_1,\Lambda_1)\Lambda_1^{l}{S}_1^{-1}(\bm{m}+j\bm{e}_2,\Lambda_1)\Lambda_2^j\Lambda_1
-\sum_{j\in\mathbb{Z}}S_2(\bm{m}+l\bm{e}_1,\Lambda_2){S}_2^{-1}(\bm{m}+(j-1)\bm{e}_1,\Lambda_2)\Lambda_1^{j}(\Delta_2^*)^{-1}\\
=&\sum_{j\in\mathbb{Z}}S_3(\bm{m}+l\bm{e}_1,\Lambda_1)\Lambda_1^{l}{S}_3^{-1}(\bm{m}+j\bm{e},\Lambda_1)\Lambda_1^{j+1}\Lambda_2^j
=\sum_{j\in\mathbb{Z}}S_3(\bm{m}+l\bm{e}_1,\Lambda_2^{-1})\Lambda_2^{-l}{S}_3^{-1}(\bm{m}+(j-1)\bm{e},\Lambda_2^{-1})\Delta_2^{-1}\Lambda_1^j\Lambda_2^j,
\end{align*}
Next by comparing coefficients of $\La_1$, we can find
\begin{align}
&\sum_{j\in\mathbb{Z}}\Big(S_{1}(\bm{m}+k\mathbf{e}_{1},\La_1)
\La_1^{l}S_1^{-1}(\bm{m}+je_{2},\La_1)\Big)_{1,[0]}\La_2^{j}
-S_{2}(\bm{m}+le_{1},\La_2)S_2^{-1}(\bm{m},\La_2)(\Delta_2^*)^{-1}\nonumber\\
=&\sum_{j\in\mathbb{Z}}(S_{3}(\bm{m}+le_{1},\La_1)\La_1^{l}S_3^{-1}(\bm{m}+je,\La_1)\La_1^{j})_{1,[0]}\La_2^{j}
=S_3(\bm{m}+l\bm{e}_1,\Lambda_2^{-1})\Lambda_2^{-l}
S_3^{-1}(\bm{m},\Lambda_2^{-1})\Lambda_2\Delta_2^{-1}.\label{la1ls2s3}
\end{align}

When $l=k>0$, \eqref{la1ls2s3} will become
\begin{align}
&\sum_{j\in\mathbb{Z}}\Big(S_{1}(\bm{m}+k\bm{e}_{1},\La_1)
\La_1^{k}S_1^{-1}(\bm{m}+j\bm{e}_{2},\La_1)\Big)_{1,[0]}\La_2^{j}\nonumber\\
=&S_{2}(\bm{m}+k\bm{e}_{1},\La_2)S_2^{-1}(\bm{m},\La_2)(\Delta_2^*)^{-1}+S_3(\bm{m}+k\bm{e}_1,\Lambda_2^{-1})\Lambda_2^{-k}
S_3^{-1}(\bm{m},\Lambda_2^{-1})\Lambda_2\Delta_2^{-1}.\label{la1ks2s3}
\end{align}
Notice that the lowest $\Lambda_2-$order in $S_{2}(\bm{m}+k\bm{e}_{1},\La_2)S_2^{-1}(\bm{m},\La_2)(\Delta_2^*)^{-1}$ is $1$, while the highest $\Lambda_2-$order in $S_3(\bm{m}+k\bm{e}_1,\Lambda_2^{-1})\Lambda_2^{-k}
S_3^{-1}(\bm{m},\Lambda_2^{-1})\Lambda_2\Delta_2^{-1}$ is $-k$, so the coefficients of $\Lambda_2^j$ ($-k<j\leq 0$) in the left hand side of \eqref{la1ks2s3} are zero. Based upon these, one can obtain the results for $\La_1^{k}(S_2(\bm{m},\La_2))$ and $\Lambda_1^k(S_3(\bm{m},\Lambda_2^{-1}))$.

When $l=-k<0$, the highest $\Lambda_2-$order in $S_{1}(\bm{m}+k\bm{e}_{1},\La_1)
\La_1^{l}S_1^{-1}(\bm{m}+j\bm{e}_{2},\La_1)$ is $-k$, therefore by \eqref{la1ls2s3} we can obtain
\begin{align*}
-S_{2}(\bm{m}-k\bm{e}_{1},\La_2)S_2^{-1}(\bm{m},\La_2)(\Delta_2^*)^{-1}=&\sum_{j\in\mathbb{Z}}\Big(S_{3}(\bm{m}-k\bm{e}_{1},\La_1)\La_1^{-k}(\La_1\La_2)^{j}
S_3^{-1}(\bm{m},\La_1)\Big)_{1,[0]}\\
=&S_3(\bm{m}-k\bm{e}_1,\Lambda_2^{-1})\Lambda_2^{k}
S_3^{-1}(\bm{m},\Lambda_2^{-1})\Lambda_2\Delta_2^{-1}.
\end{align*}
The range of $\Lambda_2$--order in above relation is $[1,k]$, thus
we can obtain $\La_1^{-k}(S_2(\bm{m},\La_2))$ and $\La_1^{-k}(S_3(\bm{m},\La_2^{-1}))$ by considering terms with $\Lambda_2$--order greater than $1$.

Similarly, we can obtain actions of $\Lambda_2^{\pm k}$ on $S_1(\bm{m},\Lambda_1)$ and $S_3(\bm{m},\Lambda_1)$.
\end{proof}

\begin{proposition}\label{prop:Hpsi}
If introduce the operator
$$H=\La_1\De_2+\rho $$
with $\rho=\frac{\widetilde{\tau}_{\bm{m}}}{\widetilde{\tau}_{\bm{m}+\bm{e}_1}}\frac{\widetilde{\tau}_{\bm{m}+\bm{e}+\bm{e}_1}}{\widetilde{\tau}_{\bm{m}+\bm{e}}}
=\p_{t_1^{(1)}}\log\frac{\widetilde{\tau}_{\bm{m}+\bm{e}}}{\widetilde{\tau}_{\bm{m}+\bm{e}_1}}$,
then $H(\Psi_i) = 0$, $i=1,2,3$.
\end{proposition}
\begin{proof}
In fact this proposition can be proved by considering $k=1$ in Proposition \ref{prop:lambda12relation} and using definitions of $\Psi_i$. Notice that by
Proposition \ref{prop:lambda12relation}, we can find
    \begin{align*}
        S_2(\bm{m}-\bm{e}_1,t,\La_2)
        =-\frac{\tau_{\bm{m}}}{\tau_{\bm{m}-\bm{e}_1}}\frac{\tau_{\bm{m}+\bm{e}_2}}{\tau_{\bm{m}+\bm{e}}}\cdot \De_2\cdot S_2.
    \end{align*}
    If further apply $\La_1$ to both sides, we have
    \begin{align}
    \De_2\cdot \La_1(S_2)+\rho S_2=0,\label{de2la1s2}
    \end{align}
    which implies $H(\Psi_2)=0$. The remaining cases are completely analogous.
\end{proof}
\begin{remark}
Relation $\frac{\widetilde{\tau}_{\bm{m}}}{\widetilde{\tau}_{\bm{m}+\bm{e}_1}}\frac{\widetilde{\tau}_{\bm{m}+\bm{e}+\bm{e}_1}}{\widetilde{\tau}_{\bm{m}+\bm{e}}}
=\p_{t_1^{(1)}}\log\frac{\widetilde{\tau}_{\bm{m}+\bm{e}}}{\widetilde{\tau}_{\bm{m}+\bm{e}_1}}$ comes from the Hirota bilinear equation of $3$--KP hierarchy
  \[
      D_1^{(1)}\tau_{\bm{m}+\bm{e}_2}\cdot\tau_{\bm{m}}=\tau_{\bm{m}+\bm{e}}\cdot \tau_{\bm{m}-\bm{e}_1},
  \]
  where $D_1^{(1)}$ is the Hirota bilinear operator with respect to $t_1^{(1)}$.
\end{remark}
Next we define the rings
\begin{alignat*}{2}
    &\E=\mathcal{B}[\La_1,\La_1^{-1},\La_2,\La_2^{-1}],\\
    &\E_{(1)}^{\pm}=\mathcal{B}[\La_2,\La_2^{-1}]((\La_1^{\mp1})), \qquad &&\E_{(1)}^{0,\pm}=\mathcal{B}((\La_1^{\mp1})),\\
    &\E_{(2)}^{\pm}=\mathcal{B}[\La_1,\La_1^{-1}]((\La_2^{\mp1})),\qquad &&\E_{(2)}^{0,\pm}=\mathcal{B}((\La_2^{\mp1})),
\end{alignat*}
where $\mathcal{B}$ is the set of the functions depending on $\bm{m}$ and $t$.
\begin{corollary}\label{corollary:HS}
    The operator multiplications of $H$ and $S_i\ (i=1,2,3)$ are given by the following identities
    \begin{align*}
        &H\cdot S_1(\bm{m},\La_1)=(\La_1-\rho(\bm{m}))\cdot S_1(\bm{m},\La_1)\cdot \De_2,\\
        &H\cdot S_2(\bm{m},\La_2)=-\rho(\bm{m}) S_2(\bm{m},\La_2)\cdot \De_1,\\
        &H\cdot S_3(\bm{m},\La_1)= (\La_1-\rho(\bm{m}))\cdot S_3(\bm{m},\La_1)\cdot \Delta_{12},\\
        &H\cdot S_3(\bm{m},\La_2^{-1})= -\rho(\bm{m}) \cdot S_3(\bm{m},\La_2^{-1})\cdot \Delta_{12}.
    \end{align*}
and
\begin{align*}
&S_1(\bm{m},\Lambda_1)\cdot \Delta_2^{-1}\cdot S_1^{-1}(\bm{m},\Lambda_1)=\iota_{\Lambda_1^{-1}}H^{-1}\cdot(\Lambda_1-\rho(\bm{m})),\\
&S_1(\bm{m},\Lambda_1)\cdot (\Delta_2^*)^{-1}\cdot S_1^{-1}(\bm{m},\Lambda_1)=-\iota_{\Lambda_1^{-1}}H^{-1}\cdot(\Lambda_1-\rho(\bm{m}))-1,\\
&S_2(\bm{m},\Lambda_2)\cdot \Delta_1^{-1}\cdot S_2^{-1}(\bm{m},\Lambda_2)=-\iota_{\Lambda_2^{-1}}H^{-1}\cdot\rho(\bm{m}),\\
&S_2(\bm{m},\Lambda_2)\cdot (\Delta_1^*)^{-1}\cdot S_2^{-1}(\bm{m},\Lambda_2)=\iota_{\Lambda_2^{-1}}H^{-1}\cdot\rho(\bm{m})-1,\\
&S_3(\bm{m},\Lambda_1)\cdot (\Delta_{12}^*)^{-1}\cdot S_3^{-1}(\bm{m},\Lambda_1)=-\iota_{\Lambda_1^{-1}}H^{-1}\cdot(\Lambda_1-\rho(\bm{m}))-1,\\
&S_3(\bm{m},\Lambda_2^{-1})\cdot \Delta_{12}^{-1}\cdot S_3^{-1}(\bm{m},\Lambda_2^{-1})=-\iota_{\Lambda_2^{-1}}H^{-1}\cdot\rho(\bm{m}).
\end{align*}
where $\iota_{\Lambda_i^{\pm1}}A$ ($i=1,2$) means expanding $A$ in $\E_{(i)}^{\mp}$ in terms of $\Lambda_i$ with coefficients belonging to $\mathcal{B}[\Lambda_{3-i},\Lambda_{3-i}^{-1}]$.
\end{corollary}
\begin{proof}
Notice that $H(\Psi_2)=0$ is equivalent to \eqref{de2la1s2}, that is, $\De_2\cdot \La_1\cdot S_2\cdot\La_1^{-1}+\rho S_2=0$.
Next after inserting $\De_2\cdot \La_1=H-\rho$ into above relation, we can obtain $H\cdot S_2$. Similarly, one can easily obtain other cases. Next $S_1\Delta_2^{-1}S_1^{-1}$ can be directly derived from $HS_1$. As for $S_1(\Delta_2^*)^{-1}S_1^{-1}$, by the result $HS_1$ we have
$S_1\Lambda_2 S_1^{-1}=\iota_{\Lambda_1^{-1}}(\Lambda_1-\rho)^{-1}H+1$, which implies $S_1\Lambda_2^{-1} S_1^{-1}=\iota_{\Lambda_1^{-1}}(\Lambda_1-\rho+H)^{-1}(\Lambda_1-\rho)$. Based upon this, we can easily obtain $S_1(\Delta_2^*)^{-1}S_1^{-1}$. Others can be derived by similar methods.
\end{proof}
\begin{lemma}\label{Lemma:operatoractonpsi}
    For any $A\in \mathcal{Q}$ with $\mathcal{Q}\in\{\E_{(1)}^{0,\pm},\E_{(2)}^{0,\pm}\}$, if $A(\Psi_{j(\mathcal{Q})})=0$, then $A=0$, where $j(\E_{(1)}^{0,+})=1$, $j(\E_{(2)}^{0,-})=2$ and $j(\E_{(1)}^{0,-})=j(\E_{(2)}^{0,+})=3$.
\end{lemma}
\begin{proof}
Firstly when $\mathcal{Q}=\E_{(1)}^{0,+}$, then by $A(\Psi_1)=0$ we can know $A\cdot S_1=0$, which implies $A=0$. Others can be similarly proved.
\end{proof}
\begin{proposition}\label{prop: sumdecom}
    The following direct sum decomposition holds,
    \begin{align*}
        \E_{(1)}^{\pm}=\E_{(1)}^{0,\pm}\oplus\E_{(1)}^{\pm}H,\qquad \E_{(2)}^{\pm}=\E_{(2)}^{0,\pm}\oplus\E_{(2)}^{\pm}H,
    \end{align*}
    where $\E_{(\al)}^{\pm}H$ is the left ideal of $\E_{(\al)}^{\pm}$ generated by $H$ for $\al =1,2$.
\end{proposition}
\begin{proof}
Here we only prove $\E_{(1)}^{+}=\E_{(1)}^{0,+}\oplus\E_{(1)}^{+}H$, since others are almost the same. Firstly given $A\in\E_{(1)}^{0,+}\oplus\E_{(1)}^{+}H$, we can know by Proposition \ref{prop:Hpsi} that $A(\Psi_1)=0$, which implies $A=0$ by Lemma \ref{Lemma:operatoractonpsi}. So $\E_{(1)}^{0,+}\oplus\E_{(1)}^{+}H$ is a direct sum. Next
 we just need to prove that
$\E_{(1)}^{+}\subseteq\E_{(1)}^{0,+}\oplus\E_{(1)}^{+}H$.
That is,  show that
\begin{align}
 \{\La_1^i\La_2^j\mid i\leq M,\ -N_1\leq j\leq N_2\}\subseteq \E_{(1)}^{0,+}\oplus\E_{(1)}^+H\label{esubseteqeoeh}
\end{align}
for some positive integers $N_1, N_2$ and $M$. Since $\La_1^i\in \E_{(1)}^0$ for $i\leq M$,
we next make induction on $j$ to complete the proof. Assuming \eqref{esubseteqeoeh} holds for $j>0$, we will prove it for $j+1$, i.e. $\La_2\cdot\La_1^i\La_2^{j}\in\E_{(1)}^0\oplus\E_{(1)}H$.
By hypothesis $\La_1^i\La_2^{j}=\sum_{l\leq N}a_l\La_1^l+FH$ for $a_l\in \mathcal{B}$ and $F\in \E_{(1)}^+$,
we have
$$\La_2\cdot\La_1^i\La_2^{j}=\sum_{l\leq N}a_l(\bm{m}+e_2)\cdot\La_1^{l-1}\cdot\big(H-\rho+\La_1\big)+\La_2\cdot FH \in\E_{(1)}^0\oplus\E_{(1)}H,
$$
where we have used $\La_2=\La_1^{-1}\cdot(H-\rho(\bm{m}))+1$.
While the case for $j<0$ is similar. So we finish the proof.
\end{proof}

Due to Proposition \ref{prop: sumdecom}, we can naturally define the following projections
$$\pi_\al^{\pm}:\E_{(\al)}^{\pm}\to \E_{(\al)}^{0,\pm}, \quad\al =1,2.$$
In order to give formulas to compute $\pi_\al^{\pm}$, we need the lemma below.
\begin{lemma}\label{Lemma:La_i}
  If denote $\pi$ as one of $\pi_{\alpha}^{\pm}$ for $\alpha=1,2$, then
  \begin{align*}
     \pi(\La_i^{\pm (k+1)})=\La_i^{\pm 1}(\pi(\La_i^{\pm k}))\cdot\pi(\La_i^{\pm 1}).
  \end{align*}
 \end{lemma}
 \begin{proof}
 By the definition of projection $\pi$, we have $\La_i^k=\pi(\La_i^k)+AH$. Multiply both sides of the equation by $\La_i$, we can get $\La_i^{k+1}=\La_i(\pi(\La_i^k))\cdot\La_i+\La_{i}AH$. After replacing $\La_i$ with $\pi(\La_i)+BH$, the formula can be written as $\La_i^{k+1}=\La_i(\pi(\La_i^k))\cdot(\pi(\La_i)+BH)+\La_{i}AH$, i.e $\La_i^{k+1}=\La_i(\pi(\La_i^k))\cdot\pi(\La_i)+CH$. The proof for $\La_i^{-k}$ is the same. So we finish the proof.
 \end{proof}
 \begin{lemma}\label{Lemma:La_12}
 Projections $\pi_i^{\pm}$ on $\Lambda_j^{\pm1}$ are given as follows
  \begin{align*}
 &\pi_1^{\pm}(\La_2)=1-\La_1^{-1}\rho(\bm{m}),\quad \pi_1^{\pm}(\La_2^{-1})=1-\iota_{\Lambda_1^{\mp1}}
 \big(\La_1-\rho(\bm{m}-\bm{e}_2)\big)^{-1}\cdot\rho(\bm{m}-\bm{e}_2),\\
 &\pi_2^\pm(\La_1)=-\iota_{\Lambda_2^{\mp1}}(\La_2-1)^{-1}\cdot\rho(\bm{m}),
 \quad\pi_2^{\pm}(\La_1^{-1})=-\rho(\bm{m}-\bm{e}_1)^{-1}\Delta_2.
  \end{align*}
 \end{lemma}
 \begin{proof}
 Notice that $\pi_1^{\pm}(\La_2)$ and $\pi_2^\pm(\La_1)$ can be directly derived from $H=\Lambda_1\Delta_2+\rho$. For $\pi_1^{\pm}(\La_2^{-1})$, it comes from $\Lambda_2^{-1}H=\Lambda_1-(\Lambda_1-\rho(\bm{m}-\bm{e}_2))\Lambda_2^{-1}$. As for $\pi_2^{\pm}(\La_1^{-1})$, we can get it by $\Lambda_1^{-1}H=\Delta_2+\rho(\bm{m}-\bm{e}_1)\Lambda_1^{-1}$.
 \end{proof}
 \begin{proposition}\label{prop:La_12 } For $k>0$,
     \begin{align*}
       &\pi_1^{\pm}(\La_2^k)=\prod_{j=1}^k\Big(1-\La_1^{-1}\rho\big(\bm{m}+(j-1)\bm{e}_2\big)\Big),\\
       &\pi_1^{\pm}(\La_2^{-k})=\prod_{j=1}^k\Big(1-\iota_{\Lambda_1^{\mp1}}
 \big(\La_1-\rho(\bm{m}-j\bm{e}_2)\big)^{-1}\cdot\rho(\bm{m}-j\bm{e}_2)\Big),\\
       &\pi_2^{\pm}(\La_1^k)=(-1)^k\prod_{j=1}^k\Big(\iota_{\Lambda_2^{\mp1}}(\La_2-1)^{-1}\cdot\rho(\bm{m}+(j-1)\bm{e}_1)\Big),\\
       &\pi_2^{\pm}(\La_1^{-k})=(-1)^k\prod_{j=1}^k\Big(\rho(\bm{m}-j\bm{e}_1)^{-1}\cdot\Delta_2\Big),
     \end{align*}
     where $\prod_{j=1}^kA_j=A_k\cdots A_2A_1$.
 \end{proposition}
\section{3--KP hierarchy by Lax operators}
In this section, we will introduce Lax operators of 3--KP hierarchy and obtain corresponding Lax equations from evolution equations of wave operators. Also we discuss evolution equations of $H$. Now let us introduce Lax operators
\begin{align*}
&L_1(\bm{m},\Lambda_1)=S_1(\bm{m},\Lambda_1)\cdot \Lambda_1\cdot S_1^{-1}(\bm{m},\Lambda_1)=\La_1+u_0^{(1)}(\bm{m})+u_1^{(1)}(\bm{m})\La_1^{-1}+\cdots,\\
&L_2(\bm{m},\Lambda_2)=S_2(\bm{m},\Lambda_2)\cdot \Lambda_2^{-1}\cdot S_2^{-1}(\bm{m},\Lambda_2)=u_{-1}^{(2)}(\bm{m})\La_2^{-1}+u_0^{(2)}(\bm{m})+u_1^{(2)}(\bm{m})\La_2+\cdots,\\
&L_3(\bm{m},\Lambda_1)=S_3(\bm{m},\Lambda_1)\cdot \Lambda_1^{-1}\cdot S_3^{-1}(\bm{m},\Lambda_1)=u_{-1}^{(3)}(\bm{m})\La_1^{-1}+u_0^{(3)}(\bm{m})+\cdots,\\
&L_3(\bm{m},\Lambda_2^{-1})=S_2(\bm{m},\Lambda_2^{-1})\cdot \Lambda_2\cdot S_2^{-1}(\bm{m},\Lambda_2^{-1})=\widetilde{u}_{1}^{(3)}(\bm{m})\La_2+\widetilde{u}_0^{(3)}(\bm{m})+\cdots.
\end{align*}
Then we have the following corollary.
\begin{corollary}
The wave functions $\Psi_i$ satisfy the following relations
\begin{align*}
&L_1(\bm{m},\Lambda_1)\Big(\Psi_1(\bm{m},z)\Big)=z\Psi_1(\bm{m},z),\quad L_2(\bm{m},\Lambda_2)\Big(\Psi_2(\bm{m},z)\Big)=z^{-1}\Psi_2(\bm{m},z),\\
&L_3(\bm{m},\Lambda_1)\Big(\Psi_3(\bm{m},z)\Big)=L_3(\bm{m},\Lambda_2^{-1})
\Big(\Psi_3(\bm{m},z)\Big)=z^{-1}\Psi_3(\bm{m},z),
\end{align*}
and for $i=1,2,3$,
\begin{align*}
&\partial_{t_k^{(1)}}\Psi_i(\bm{m},z)=B_k^{(1)}(\bm{m},\Lambda_1)\Big(\Psi_i(\bm{m},z)\Big),
\quad \partial_{t_k^{(2)}}\Psi_i(\bm{m},z)=B_k^{(2)}(\bm{m},\Lambda_2)\Big(\Psi_i(\bm{m},z)\Big),\\
&\partial_{t_k^{(3)}}\Psi_i(\bm{m},z)=B_k^{(3)}(\bm{m},\Lambda_1)\Big(\Psi_i(\bm{m},z)\Big)
=B_k^{(3)}(\bm{m},\Lambda_2^{-1})\Big(\Psi_i(\bm{m},z)\Big),
\end{align*}
where $B_k^{(j)}$ ($j=1,2,3$) are defined by
\begin{align*}
&B_k^{(1)}(\bm{m},\Lambda_1)=\Big(L_1^k(\bm{m},\Lambda_1)\Big)_{1,\geq 0},\quad
B_k^{(2)}(\bm{m},\Lambda_2)=\Big(L_2^k(\bm{m},\Lambda_2)\Big)_{\Delta_2^*,\geq 1},\\
&B_k^{(3)}(\bm{m},\Lambda_1)=\Big(L_3^k(\bm{m},\Lambda_1)\Big)_{1,<0},\quad B_k^{(3)}(\bm{m},\Lambda_2^{-1})=\Big(L_2^k(\bm{m},\Lambda_2^{-1})\Big)_{\Delta_2,\geq 1}.
\end{align*}

\end{corollary}
\begin{proof}
Firstly the actions of $L_i$ on $\Psi_i$ are obvious by corresponding definitions. Then for $\partial_{t_k^{(i)}}\Psi_i(\bm{m},z)\ (i=1,2,3)$ and $\partial_{t_k^{(j)}}\Psi_3(\bm{m},z)\ (j=1,2)$, they can be obtained directly by $\partial_{t_k^{(i)}}S_i$ and $\partial_{t_k^{(j)}}S_3$ in Proposition \ref{prop:evolutionS} and definitions of $\Psi_i$. As for $\partial_{t_k^{(1)}}\Psi_2(\bm{m},z)$, it comes from the facts below
\begin{align}
\partial_{t_k^{(1)}}S_2(\bm{m},\Lambda_2)
=B_k^{(1)}(\bm{m},\Lambda_1)\Big(S_2(\bm{m},\Lambda_2)\Big),\label{patk1s2}
\end{align}
which is derived by $\partial_{t_k^{(1)}}S_2(\bm{m},\Lambda_2)$ in Proposition \ref{prop:evolutionS} and $\Lambda_1^{l}(S_2)$ ($l> 0$) in Proposition \ref{prop:lambda12relation}. Similarly, we can obtain the results for $\partial_{t_k^{(i)}}\Psi_j(\bm{m},z)$ with $i\neq j$ and $j\neq 3$.
\end{proof}
\begin{theorem}
Lax operators $L_i$ of 3$-$ KP hierarchy satisfy the following Lax equations
\begin{itemize}
  \item $t_k^{(1)}$--flow
  \begin{align*}
&\partial_{t_k^{(1)}}L_1(\Lambda_1)=[B_k^{(1)}(\Lambda_1),L_1(\Lambda_1)],\quad
\partial_{t_k^{(1)}}L_2(\Lambda_2)=[\pi_2^-(B_k^{(1)}(\Lambda_1)),L_2(\Lambda_2)],\\
&\partial_{t_k^{(1)}}L_3(\Lambda_1)=[B_k^{(1)}(\Lambda_1),L_3(\Lambda_1)],\quad
\partial_{t_k^{(1)}}L_3(\Lambda_2^{-1})=[\pi_2^+(B_k^{(1)}(\Lambda_1)),L_3(\Lambda_2^{-1})].
\end{align*}
  \item $t_k^{(2)}$--flow
  \begin{align*}
&\partial_{t_k^{(2)}}L_1(\Lambda_1)=[\pi_1^+(B_k^{(2)}(\Lambda_1)),L_1(\Lambda_1)],\quad
\partial_{t_k^{(2)}}L_2(\Lambda_2)=[B_k^{(2)}(\Lambda_2),L_2(\Lambda_2)],\\
&\partial_{t_k^{(2)}}L_3(\Lambda_1)=[\pi_1^-(B_k^{(2)}(\Lambda_2)),L_3(\Lambda_1)],\quad
\partial_{t_k^{(2)}}L_3(\Lambda_2^{-1})=[B_k^{(2)}(\Lambda_2),L_3(\Lambda_2^{-1})].
\end{align*}
  \item $t_k^{(3)}$--flow
  \begin{itemize}
    \item $\Lambda_1$--operator
    \begin{align*}
&\partial_{t_k^{(3)}}L_1(\Lambda_1)=[B_k^{(3)}(\Lambda_1),L_1(\Lambda_1)],\quad
\partial_{t_k^{(3)}}L_2(\Lambda_2)=[\pi_2^-(B_k^{(3)}(\Lambda_1)),L_2(\Lambda_2)],\\
&\partial_{t_k^{(3)}}L_3(\Lambda_1)=[B_k^{(3)}(\Lambda_1),L_3(\Lambda_1)],\quad
\partial_{t_k^{(3)}}L_3(\Lambda_2^{-1})=[\pi_2^+(B_k^{(3)}(\Lambda_1)),L_3(\Lambda_2^{-1})]
\end{align*}
    \item $\Lambda_2$--operator
\begin{align*}
&\partial_{t_k^{(3)}}L_1(\Lambda_1)=[\pi_1^+(B_k^{(3)}(\Lambda_2^{-1})),L_1(\Lambda_1)],\quad
\partial_{t_k^{(3)}}L_2(\Lambda_2)=[B_k^{(3)}(\Lambda_2^{-1}),L_2(\Lambda_2)],\\
&\partial_{t_k^{(3)}}L_3(\Lambda_1)=[\pi_1^-(B_k^{(3)}(\Lambda_2^{-1})),L_3(\Lambda_1)],\quad
\partial_{t_k^{(3)}}L_3(\Lambda_2^{-1})=[B_k^{(3)}(\Lambda_2^{-1}),L_3(\Lambda_2^{-1})]
\end{align*}
  \end{itemize}
\end{itemize}
\end{theorem}
\begin{proof}
Firstly $\partial_{t_k^{(1)}}L_1(\Lambda_1)$ can be directly obtained by $\partial_{t_k^{(1)}}S_1(\Lambda_1)$ in Proposition \ref{prop:evolutionS} and $L_1(\Lambda_1)=S_1(\Lambda_1)\cdot\Lambda_1\cdot S_1^{-1}(\Lambda_1)$. As for
$\partial_{t_k^{(1)}}L_2(\Lambda_2)$, one can firstly find by $L_2(\Lambda_2)=S_2(\Lambda_2)\Lambda_2^{-1}S_2^{-1}(\Lambda_2)$ that
\begin{align*}
\partial_{t_k^{(1)}}L_2(\Lambda_2)=\left[\partial_{t_k^{(2)}}S_2(\Lambda_2)\cdot S_2^{-1}(\Lambda_2),L_2(\Lambda_2)\right].
\end{align*}
On the other hand, we can find by Proposition \ref{prop:Hpsi} that $B_k^{(1)}(\Psi_2)=\pi_{2}^-(B_k^{(1)})(\Psi_2)$. Then by Lemma \ref{Lemma:operatoractonpsi} and definition of $\Psi_2$, we can obtain
\begin{align*}
\pi_2^-(B_k^{(1)})=B_k^{(1)}(S_2)\cdot S_2^{-1},
\end{align*}
which leads to the result of $\partial_{t_k^{(1)}}L_2(\Lambda_2)$ by (\ref{patk1s2}). Other cases can be similarly derived.
\end{proof}

We next compute the derivatives of $H$. Before doing that, let us do the following preparation.

\begin{lemma}\label{lemma:ABH}
    If $A\in\E$ satisfies $A(\Psi_\al)=0$ for $1\leq \al\leq 3$, then there is a unique operator $B\in\E$ such that $A=BH$.
\end{lemma}
\begin{proof}
Firstly notice that  $A\in\E\subseteq\E_{(1)}^{+}$, then by Proposition \ref{prop: sumdecom}, there exist unique $\widetilde{A}\in\E_{(1)}^{0,+}$ and $B\in \E_{(1)}^{+}$ such that $A=\widetilde{A}+BH$. Next according to Proposition \ref{prop:Hpsi}, we can know $\widetilde{A}(\Psi_1)=0$, which implies $\widetilde{A}=0$ by Lemma \ref{Lemma:operatoractonpsi}. Therefore $A=BH$ with $B\in\E_{(1)}^{+}$. So if denote $$B=\sum_{j\leq M}\sum_{i=N_j}^{K_j}b_{ij}\Lambda_1^i\Lambda_2^j,$$
then by $BH=A\in\mathcal{E}$, we can know the lowest $\Lambda_1$--order in $B$ should be finite. Therefore $B\in\mathcal{E}$.
\end{proof}
By Corollary \ref{corollary:HS}, we can obtain the lemma below.
\begin{lemma}\label{corollary:HL}
    \begin{align*}
        &H\cdot L_1(\La_1)=(\La_1-\rho)\cdot L_1(\La_1) \cdot\iota_{\La_1^{-1}}(\La_1-\rho)^{-1}H,\\
        &H\cdot L_2(\La_2)=\rho\cdot L_1(\La_2)\cdot \rho^{-1}H,\\
        &H\cdot L_3(\La_1)=(\La_1-\rho)\cdot L_3(\La_1)\cdot \iota_{\La_1}(\La_1-\rho)^{-1}H,\\
        &H\cdot L_3(\La_2^{-1})=\rho\cdot L_3(\La_2^{-1})\cdot \rho^{-1}H.
    \end{align*}
\end{lemma}
\begin{proposition}
Evolution equation of $H$ is given by
\begin{align*}
\p_{t_k^{(i)}}H=C^{(i)}H-H\cdot B_k^{(i)},\quad i=1,2,3,
\end{align*}
where $C^{(i)}$ is given by
    \begin{align*}
        &C^{(1)}=\Big((\La_1-\rho)\cdot L_1^k(\La_1) \cdot \iota_{\Lambda_1^{-1}}(\La_1-\rho)^{-1} \Big)_{1,\geq 0},\\
        &C^{(2)}=\Big(\rho \cdot L_2^k(\La_2) \cdot\rho^{-1}\Big)_{2,\leq 0}-\Big((\De_2^*)^{-1}\La_2^{-1}\rho \cdot L_2^k(\La_2) \cdot\rho^{-1}\Big)_{2,[0]},\\
        &C^{(3)}(\La_1)=\Big((\La_1-\rho)\cdot L_3^k(\La_1) \cdot \iota_{\Lambda_1}(\La_1-\rho)^{-1} \Big)_{1,< 0},\\
        &C^{(3)}(\La_2)=\Big(\rho \cdot L_3^k(\La_2^{-1}) \cdot\rho^{-1}\Big)_{2,\geq 0}-\Big(\De_2^{-1}\La_2\rho\cdot L_3^k(\La_2^{-1}) \cdot\rho^{-1}\Big)_{2,[0]}.
    \end{align*}
\end{proposition}
\begin{proof}
Firstly apply $\partial_{t_k^{(i)}}$ to $H(\Psi_\alpha)=0$ for $\alpha=1,2,3$ and use Lemma \ref{lemma:ABH}, then we can find there exists $C^{(i)}\in\mathcal{E}$ such that
\begin{align*}
\partial_{t_k^{(i)}}H=C^{(i)}H-H\cdot B_k^{(i)}.
\end{align*}
Notice that $\partial_{t_k^{(i)}}H$ is just a function, thus $C^{(i)}H-H\cdot B_k^{(i)}$ must be also a function and $C^{(i)}$ is uniquely determined by this property.

For $i=1$, from Lemma \ref{corollary:HL},
$$H\cdot L_1^k=(\La_1-\rho)\cdot L_1^k \cdot \iota_{\Lambda_1^{-1}}(\La_1-\rho)^{-1}\cdot H,$$
which implies that
\begin{align}
\big((\La_1-\rho)\cdot L_1^k \cdot \iota_{\Lambda_1^{-1}}(\La_1-\rho)^{-1}\big)_{1,\geq0}H-H(L_1^k)_{1,\geq 0}
        =H(L_1^k)_{1,< 0}-\big((\La_1-\rho)\cdot L_1^k \cdot \iota_{\Lambda_1^{-1}}(\La_1-\rho)^{-1}\big)_{1,<0}H,\label{HL1pnrelation}
\end{align}
Notice that for coefficients of $\Lambda_2$, LHS of \eqref{HL1pnrelation} has non--negative $\Lambda_1$--order, while RHS of \eqref{HL1pnrelation} owns negative $\Lambda_1$--order. So coefficients of $\Lambda_2$ in LHS of \eqref{HL1pnrelation} must be zero. Similarly, we can prove that coefficients of $\Lambda_2^0$ in LHS of \eqref{HL1pnrelation} is just a function. Thus we can find $\big((\La_1-\rho)\cdot L_1^k \cdot \iota_{\Lambda_1^{-1}}(\La_1-\rho)^{-1}\big)_{1,\geq0}$ satisfying the required condition, which means $C^{(1)}=\big((\La_1-\rho)\cdot L_1^k \cdot \iota_{\Lambda_1^{-1}}(\La_1-\rho)^{-1}\big)_{1,\geq0}$.

    For $i=2$, similarly we have
    \[
        H(L_2^k)_{\De_2^*,\geq 1}-(\rho L_2^k \rho^{-1})_{\De_2^*,\geq 0}H=(\rho L_2^k \rho^{-1})_{\De_2^*,\leq -1}H-H(L_2^k)_{\De_2^*,\leq 0}.
    \]
    Considering the coefficients of $\La_1$ and $\La_1^0$ respectively, it gives
    \[
        \La_1\De_2(L_2^k(\De_2^*)^{-1})_{2,\leq 1}+(\rho L_2^k \rho^{-1})_{2,\leq 1}\La_1\La_2
        =-(\rho L_2^k \rho^{-1})_{2,\geq 1}\La_1\La_2-\La_1\De_2(L_2^k(\De_2^*)^{-1})_{2,\geq 1},
    \]
    \[
        \rho(L_2^k(\De_2^*)^{-1})_{2,\leq 0}\De_2^*-(\rho L_2^k \rho^{-1})_{2,\leq 0}\rho
        =(\rho L_2^k \rho^{-1})_{2,\geq 0}\rho-\rho(L_2^k(\De_2^*)^{-1})_{2,\geq 0}\De_2^*.
    \]
    Further by comparing the $\La_2-$powers for both equations, we have
    \[
        \La_1\De_2(L_2^k(\De_2^*)^{-1})_{2,\leq 0}+(\rho L_2^k \rho^{-1})_{2,\leq 0}\La_1\La_2
        -\La_1(L_2^k(\De_2^*)^{-1})_{2,[1]}=0,
    \]
    \[
        \rho(L_2^k(\De_2^*)^{-1})_{2,\leq 0}\De_2^*-(\rho L_2^k \rho^{-1})_{2,\leq 0}\rho
        +\rho(L_2^k(\De_2^*)^{-1})_{2,[1]}\La_2^{-1}=0,
    \]
    which implies
    \begin{align*}
        &\left((\rho L_2^k \rho^{-1})_{2,\leq 0}-\big((\De_2^*)^{-1}\La_2^{-1}\rho L_2^k \rho^{-1}\big)_{2,[0]}\right)H
        -H\cdot(L_2^k(\De_2^*)^{-1})_{2,\leq 0}\De_2^*\\
        &=\left(\rho L_2^k(\De_2^*)^{-1}\La_2^{-1}-(\De_2^*)^{-1}\La_2^{-1}\rho L_2^k\right)_{2,[0]}.
    \end{align*}
    The RHS is a function satisfying required condition, so we have
    $$C^{(2)}=(\rho L_2^k \rho^{-1})_{2,\leq 0}-\big((\De_2^*)^{-1}\La_2^{-1}\rho L_2^k \rho^{-1}\big)_{2,[0]},$$

    The remaining $C^{(3)}(\La_i)$ has the same proof with $C^{(i)}$ for $i=1,2$.
\end{proof}
\noindent\textbf{Example. }
Firstly note that
\begin{align*}
  &B_1^{(1)}(\bm{m},\La_1)= \Lambda_1+u_0^{(1)}(\bm{m}),\\
  &\pi_2^+\Big(B_1^{(1)}(\bm{m},\La_1)\Big)= -\De_2^{-1}\rho(\bm{m})+u_0^{(1)}(\bm{m}),  \\
  &\pi_2^-\Big(B_1^{(1)}(\bm{m},\La_1)\Big)= \Lambda_2^{-1}(\De_2^{*})^{-1}\rho(\bm{m})+u_0^{(1)}(\bm{m}).
\end{align*}
Then one can find that
\begin{align*}
  &\partial_{t_{k}^{(1)}}u_0^{(1)}(\bm{m})=u_1^{(1)}(\bm{m}+e_1)-u_1^{(1)}(\bm{m}),\\
  &\partial_{t_{k}^{(1)}}u_{-1}^{(2)}(\bm{m})=u_0^{(1)}(\bm{m})u_{-1}^{(2)}(\bm{m})-u_{-1}^{(2)}u_0^{(1)}(\bm{m}-\bm{e}_2),\\
  &\partial_{t_{k}^{(1)}}u_{-1}^{(3)}(\bm{m})=u_0^{(1)}(\bm{m})u_{-1}^{(3)}(\bm{m})-u_{-1}^{(3)}u_0^{(1)}(\bm{m}-\bm{e}_1) .
\end{align*}

\section{Lax operator of $[n_1,n_2,n_3]$--KdV hierarchy}
In this section, we will construct Lax operator of $[n_1,n_2,n_3]$--KdV hierarchy from the corresponding bilinear equation.
\begin{theorem}
For $[n_1,n_2,n_3]$--KdV hierarchy, if denote operators $\mathcal{L},\widetilde{\mathcal{L}}\in\E$ as follows,
\begin{align*}
&\mathcal{L}(\Lambda_1,\Lambda_2)=B_{n_1}^{(1)}(\Lambda_1)+B_{n_2}^{(2)}(\Lambda_2)
+B_{n_3}^{(3)}(\Lambda_1),\\
&\widetilde{\mathcal{L}}(\Lambda_1,\Lambda_2)=B_{n_1}^{(1)}(\Lambda_1)+B_{n_2}^{(2)}(\Lambda_2)+B_{n_3}^{(3)}(\Lambda_2^{-1}).
\end{align*}
then we can find
\begin{align*}
&\mathcal{L}(\Psi_1)=z^{n_1}\Psi_1,\quad \mathcal{L}(\Psi_2)=z^{-n_2}\Psi_2,\quad \mathcal{L}(\Psi_3)=z^{-n_3}\Psi_3,\\
&\widetilde{\mathcal{L}}(\Psi_1)=z^{n_1}\Psi_1,\quad \widetilde{\mathcal{L}}(\Psi_2)=z^{-n_2}\Psi_2,\quad \widetilde{\mathcal{L}}(\Psi_3)=z^{-n_3}\Psi_3.
\end{align*}
\end{theorem}
\begin{proof}
Here we only prove the case of $\mathcal{L}$, since $\widetilde{\mathcal{L}}$ can be similarly done. Firstly set $l=1$ and $t'=t$ in Proposition \ref{bilinear-wave-operator}, then we can obtain
\begin{align}
&\sum_{j\in\mathbb{Z}}S_1(\bm{m},\Lambda_1)\Lambda_1^{n_1}
S_1^{-1}(\bm{m}+j\bm{e}_2,\Lambda_1)\Lambda_1\Lambda_2^{j}
-\sum_{j\in\mathbb{Z}}S_2(\bm{m},\Lambda_2)\Lambda_2^{-n_2}
{S}^{-1}_2(\bm{m}+(j-1)\bm{e}_1,\Lambda_2)(\Delta_2^*)^{-1}\Lambda_1^{j}\nonumber\\
=&\sum_{j\in\mathbb{Z}}S_3(\bm{m},\Lambda_1)\Lambda_1^{-n_3}
{S}^{-1}_3(\bm{m}+j\bm{e},\Lambda_1)\Lambda_1^{j+1}\Lambda_2^{j}.
\label{n1n2n3kpbilinear}
\end{align}
Notice that coefficients of $\Lambda_2^0$ in \eqref{n1n2n3kpbilinear} give rise to
\begin{align}
L_1^{n_1}(\Lambda_1)
-\left(L_2^{n_2}(\Lambda_2)\cdot S_2(\Lambda_2)\cdot\sum_{j\in\mathbb{Z}}\Lambda_1^{j}
\cdot{S}^{-1}_2(\Lambda_2)\cdot(\Delta_2^*)^{-1}\right)_{2,[0]}=L_3^{n_3}(\Lambda_1).\label{L1n1L3n3}
\end{align}

The negative $\Lambda_1$--orders of \eqref{L1n1L3n3} imply that
\begin{align}
\left(L_1^{n_1}(\Lambda_1)\right)_{1,<0}
-\left(L_2^{n_2}(\Lambda_2)\cdot S_2(\Lambda_2)\cdot\Delta_1^{-1}
\cdot {S}^{-1}_2(\Lambda_2)\cdot(\Delta_2^*)^{-1}\right)_{2,[0]}
=\left(L_3^{n_3}(\Lambda_1)\right)_{1,<0}.\label{L1n1negative0}
\end{align}
If assume $L_2^{n_2}=\sum_{i\leq n_2} a_i \Lambda_2^{-i}$, then by Proposition \ref{prop:lambda12relation} and Lemma \ref{lemma:Alade}
\begin{align*}
&\left(L_2^{n_2}\cdot S_2\cdot\Delta_1^{-1}
\cdot {S}^{-1}_2\cdot(\Delta_2^*)^{-1}\right)_{2,[0]}=\sum_{i=1}^{n_2}a_i \left(\Lambda_2^{-i}(S_1)\cdot S_1^{-1}-1\right)\\
&=\left((L_2^{n_2})_{2,<0}
-(L_2^{n_2})_{2,<0}(1)\right)(S_1)\cdot S_1^{-1}=B_{n_2}^{(2)}(S_1)\cdot S_1^{-1}.
\end{align*}
Therefore we can find \eqref{L1n1negative0} becomes into
\begin{align*}
\left(L_1^{n_1}(\Lambda_1)\right)_{1,<0}
=B_{n_2}^{(2)}(\Lambda_1)(S_1(\Lambda_1))\cdot S_1^{-1}(\Lambda_1)
+B_{n_3}^{(3)}(\Lambda_1),
\end{align*}
which implies $\left(L_1^{n_1}(\Lambda_1)\right)_{1,<0}(\Psi_1)=B_{n_2}^{(2)}(\Lambda_1)(\Psi_1)
+B_{n_3}^{(3)}(\Lambda_1)(\Psi_1)$. Therefore $\mathcal{L}(\Psi_1)=z^{n_1}\Psi_1$.

By considering non--negative $\Lambda_1$--orders of \eqref{L1n1L3n3}, we have
\begin{align}
B_{n_1}^{(1)}(\Lambda_1)
-\left(L_2^{n_2}(\Lambda_2)\cdot S_2(\Lambda_2)\cdot(1+(\Delta_1^*)^{-1})
\cdot{S}^{-1}_2(\Lambda_2)\cdot(\Delta_2^*)^{-1}\right)_{2,[0]}=L_3^{n_3}(\Lambda_1)_{1,\geq 0}.\label{L3n3positive}
\end{align}
Recall that $L_2^{n_2}=\sum_{i\leq n_2} a_i \Lambda_2^{-i}$, thus we can obtain
\begin{align*}
&\left(L_2^{n_2}(\Lambda_2)\cdot S_2(\Lambda_2)\cdot(1+(\Delta_1^*)^{-1})
\cdot{S}^{-1}_2(\Lambda_2)\cdot(\Delta_2^*)^{-1}\right)_{2,[0]}(\Psi_3)\\
=&\sum_{i=1}^{n_2} a_i\Psi_3-\sum_{i=1}^{n_2}\left(a_i\Lambda_2^{-i}\cdot S_3(\Lambda_1)\cdot\Lambda_1^{i}\Lambda_2^{i}\right)\left(z^{m_1-m_2}\right)\cdot e^{\xi(t_3,z^{-1})}\\
=&\sum_{i=1}^{n_2} a_i\Psi_3-\sum_{i=1}^{n_2}a_i\Lambda_2^{-i}(\Psi_3)
=B_{n_2}^{(2)}(\Lambda_2)(\Psi_3).
\end{align*}
Therefore by \eqref{L3n3positive}, we can know $\mathcal{L}(\Psi_3)=z^{-n_3}\Psi_3$.

As for $\mathcal{L}(\Psi_2)$, we can consider coefficients of $\Lambda_1$ in
\eqref{n1n2n3kpbilinear}, that is,
\begin{align*}
&L_2^{n_2}(\Lambda_2)\cdot (\Delta_2^*)^{-1}=\left(S_1(\Lambda_1)\Lambda_1^{n_1}
\sum_{j\in\mathbb{Z}} \Lambda_2^{j}S_1^{-1}(\Lambda_1)\right)_{1,[0]}
-\left(S_3(\Lambda_1)\Lambda_1^{-n_3}\sum_{j\in\mathbb{Z}}\Lambda_1^{j}\Lambda_2^{j}
{S}^{-1}_3(\Lambda_1)\right)_{1,[0]},
\end{align*}
which implies by Lemma \ref{lemma:Alade} that
\begin{align*}
\left(L_2^{n_2}(\Lambda_2)\right)_{\Delta_2^*,\leq0}=&\left(L_2^{n_2}(\Lambda_2)\cdot (\Delta_2^*)^{-1}\right)_{2,>0}\Delta_2^*\\
=&\left(S_1(\Lambda_1)\Lambda_1^{n_1}
(\Delta_2^*)^{-1} S_1^{-1}(\Lambda_1)\right)_{1,[0]}\Delta_2^*
-\left(S_3(\Lambda_1)\Lambda_1^{-n_3}(\Delta_{12}^*)^{-1}
{S}^{-1}_3(\Lambda_1)\right)_{1,[0]}\Delta_2^*.
\end{align*}
Then based upon this, we can finally get $\mathcal{L}(\Psi_2)=z^{-n_2}\Psi_2$ by Proposition \ref{prop:lambda12relation}.
\end{proof}

Note that if we apply $\mathcal{L}$ or $\widetilde{\mathcal{L}}$ to both sides of 3--KP bilinear equation \eqref{3KPwavebilinear} for $l=0$, then we can recover \eqref{3KPwavebilinear} for $l=1$. And successive application will give \eqref{3KPwavebilinear} for general $l\geq 1$. Therefore we can further obtain the corollary below.
\begin{corollary}\label{corollary:n1n2n3KP}
Bilinear equation \eqref{3KPwavebilinear} of  $[n_1,n_2,n_3]$--KdV hierarchy is equivalent to the following two points:
\begin{itemize}
  \item 3--KP bilinear equation
  \begin{align*}
&\oint_{C_R}\frac{dz}{2\pi iz}\Psi_{1}(\bm{m},t,z)\widetilde{\Psi}_{1}(\bm{m}',t',z)\nonumber\\
=&\oint_{C_r}\frac{dz}{2\pi iz}\left(\Psi_{2}(\bm{m},t,z)\widetilde{\Psi}_{2}(\bm{m}',t',z)
+\Psi_{3}(\bm{m},t,z)\widetilde{\Psi}_{3}(\bm{m}',t',z)\right).
\end{align*}
  \item there exists $\mathfrak{L}\in\E$ such that $\mathfrak{L}(\Psi_1)=z^{n_1}\Psi_1$, $\mathfrak{L}(\Psi_2)=z^{-n_2}\Psi_2$,
      $\mathfrak{L}(\Psi_3)=z^{-n_3}\Psi_3$.
\end{itemize}
\end{corollary}
Notice that $\mathcal{L}$ or $\widetilde{\mathcal{L}}$ satisfies the second point in Corollary \ref{corollary:n1n2n3KP}, which is called the Lax operator of $[n_1,n_2,n_3]$--KdV hierarchy.
If we use $(L_1(\Lambda_1),L_2(\Lambda_2),L_3(\Lambda_1))$ to describe 3--KP hierarchy, we use $\mathcal{L}$ as $[n_1,n_2,n_3]$--KdV Lax operator,
while when 3--KP is expressed by $(L_1(\Lambda_1),L_2(\Lambda_2),L_3(\Lambda_2^{-1}))$,  the corresponding Lax operator is $\widetilde{\mathcal{L}}$.
\begin{corollary}
$[n_1,n_2,n_3]$--KdV Lax operator $\mathcal{L}$ and $\widetilde{\mathcal{L}}$ satisfy
\begin{align*}
&\pi_1^+(\mathcal{L})=\pi_1^+(\widetilde{\mathcal{L}})=L_1^{n_1}(\Lambda_1),\\
&\pi_1^-(\mathcal{L})=\pi_1^-(\widetilde{\mathcal{L}})=L_3^{n_3}(\Lambda_1),\\
&\pi_2^+(\mathcal{L})=\pi_2^+(\widetilde{\mathcal{L}})=L_3^{n_3}(\Lambda_2^{-1}),\\
&\pi_2^-(\mathcal{L})=\pi_2^-(\widetilde{\mathcal{L}})=L_2^{n_2}(\Lambda_2).
\end{align*}
\end{corollary}
Finally we give the explicit forms of $\mathcal{L}$ and $\widetilde{\mathcal{L}}$, that is,
\begin{align*}
\mathcal{L}=\Lambda_1^{n_1}+\sum_{k=-n_3}^{n_1-1}u_k\Lambda_1^k+\sum_{k=1}^{n_2}v_k(\La_2^{-k}-1),
\quad\widetilde{\mathcal{L}}=\Lambda_1^{n_1}+\sum_{k=0}^{n_1-1}u_k\Lambda_1^k
+\sum_{k=1}^{n_2}v_k(\La_2^{-k}-1)+\sum_{k=1}^{n_3}\widetilde{v}_k(\La_2^{k}-1).
\end{align*}
\begin{itemize}
  \item in $\mathcal{E}_{(1)}^{\pm}$
  \begin{align*}
\pi_1^{\pm}(\mathcal{L})=&\Lambda_1^{n_1}+\sum_{k=-n_3}^{n_1-1}u_k\Lambda_1^k
  +\sum_{k=1}^{n_2}v_k\left(\prod_{j=1}^k\Big(1-\iota_{\Lambda_1^{\mp1}}
 \big(\La_1-\rho(\bm{m}-j\bm{e}_2)\big)^{-1}\cdot\rho(\bm{m}-j\bm{e}_2)\Big)-1\right),\\
\pi_1^{\pm}(\widetilde{\mathcal{L}})=&\Lambda_1^{n_1}+\sum_{k=0}^{n_1-1}u_k\Lambda_1^k
+\sum_{k=1}^{n_2}v_k\left(\prod_{j=1}^k\Big(1-\iota_{\Lambda_1^{\mp1}}
 \big(\La_1-\rho(\bm{m}-j\bm{e}_2)\big)^{-1}\cdot\rho(\bm{m}-j\bm{e}_2)\Big)-1\right)\\
&+\sum_{k=1}^{n_3}\widetilde{v}_k
\left(\prod_{j=1}^k\Big(1-\La_1^{-1}\rho\big(\bm{m}+(j-1)\bm{e}_2\big)\Big)-1\right).
  \end{align*}
\item in $\mathcal{E}_{(2)}^{\pm}$
  \begin{align*}
\pi_2^{\pm}(\mathcal{L})=&(-1)^{n_1}\prod_{j=1}^{n_1}\Big(\iota_{\Lambda_2^{\mp1}}(\La_2-1)^{-1}\cdot\rho(\bm{m}+(j-1)\bm{e}_1)\Big)\\
&+\sum_{k=1}^{n_1-1}(-1)^k u_k\prod_{j=1}^k\Big(\iota_{\Lambda_2^{\mp1}}(\La_2-1)^{-1}\cdot\rho(\bm{m}+(j-1)\bm{e}_1)\Big)\\
&+\sum_{k=1}^{n_3}(-1)^k u_{-k}\prod_{j=1}^k\Big(\rho(\bm{m}-j\bm{e}_1)^{-1}\cdot\Delta_2\Big)+u_0+\sum_{k=1}^{n_2}v_k(\La_2^{-k}-1),\\
\pi_2^{\pm}(\widetilde{\mathcal{L}})=&(-1)^{n_1}\prod_{j=1}^{n_1}\Big(\iota_{\Lambda_2^{\mp1}}(\La_2-1)^{-1}\cdot\rho(\bm{m}+(j-1)\bm{e}_1)\Big)\\
&+\sum_{k=1}^{n_1-1}(-1)^k u_k\prod_{j=1}^k\Big(\iota_{\Lambda_2^{\mp1}}(\La_2-1)^{-1}\cdot\rho(\bm{m}+(j-1)\bm{e}_1)\Big)\\
&+u_0+\sum_{k=1}^{n_2}v_k(\La_2^{-k}-1)+\sum_{k=1}^{n_3}\widetilde{v}_k(\La_2^{k}-1).
  \end{align*}
\end{itemize}
\section{Conclusions and Discussions}
Here we have succeeded in finding Lax formulations of 3--KP hierarchy and its reduction $[n_1,n_2,n_3]$--KdV hierarchy by Shiota methods. For 3--KP,
\begin{itemize}
  \item expressed by Lax triple $(L_1(\Lambda_1),L_2(\Lambda_2),L_3(\Lambda_1))$ or $(L_1(\Lambda_1),L_2(\Lambda_2),L_3(\Lambda_2^{-1}))$.
  \item $(L_1(\Lambda_1),L_3(\Lambda_1))$ satisfies 2--Toda hierarchy.
  \item $(L_3(\Lambda_2^{-1}),L_2(\Lambda_2))$ satisfies 2--modified Toda hierarchy\cite{Guan2024}.
  \item use operator $H=\Lambda_1\Delta_2+\rho$ to relate $\Lambda_1$ and $\Lambda_2$.
\end{itemize}
For $[n_1,n_2,n_3]$--KdV, the corresponding Lax operator is given by
\begin{align*}
&\mathcal{L}(\Lambda_1,\Lambda_2)=B_{n_1}^{(1)}(\Lambda_1)+B_{n_2}^{(2)}(\Lambda_2)
+B_{n_3}^{(3)}(\Lambda_1),\\
&\widetilde{\mathcal{L}}(\Lambda_1,\Lambda_2)=B_{n_1}^{(1)}(\Lambda_1)+B_{n_2}^{(2)}(\Lambda_2)+B_{n_3}^{(3)}(\Lambda_2^{-1}).
\end{align*}
Since 3--KP and its reduction comes from the infinite dimensional Lie algebras, which means there are infinite symmetries for 3--KP, the flows $t_k^{(i)}$ with $i=1,2,3$ can commute with each other. One can use similar methods in \cite{ChengJP2021} to prove commutativity of times flows.
We believe results here for 3--component KP theory can be generalized to the general multi--component case.\\

\noindent{\bf Acknowledgements}:

This work is supported by National Natural Science Foundation of China (Grant Nos. 12171472 and 12261072)
and ``Qinglan Project'' of Jiangsu Universities.\\

\noindent{\bf Conflict of Interest}:

 The authors have no conflicts to disclose.\\

\noindent{\bf Data availability}:

Date sharing is not applicable to this article as no new data were created or analyzed in this study.

\end{document}